\documentclass[12pt]{article}
\usepackage{fullpage}
\usepackage{times}

\usepackage[bookmarks=false]{hyperref}

\usepackage{doi}

\usepackage{amsthm,amsmath,amssymb}
\usepackage[dvipsnames]{xcolor}
\usepackage{proof}
\usepackage{graphicx}
\usepackage{relsize}
 \usepackage{float}
\usepackage{exscale}
\usepackage{ebproof}
\usepackage{hyperref}

\usepackage{verbatim}
\usepackage{stmaryrd}
\usepackage{multicol}
\usepackage{appendix}
\usepackage{etex}
\usepackage{etoolbox}
\usepackage{xspace}
\usepackage{amsthm}

\theoremstyle{plain}
\newtheorem{theorem}{Theorem}
\newtheorem{lemma}[theorem]{Lemma}
\newtheorem{proposition}[theorem]{Proposition}
\newtheorem{corollary}[theorem]{Corollary}

\theoremstyle{definition}
\newtheorem{definition}{Definition}
\theoremstyle{remark}

\newcommand{\dL}{dL\xspace}
\newcommand{\dTL}{dTL\xspace}

\newcommand{\bebecomes}{\mathrel{::=}}

\newcommand\Vtextvisiblespace[1][.3em]{%
 \mbox{\kern.06em\vrule height.3ex}%
 \vbox{\hrule width#1}%
 \hbox{\vrule height.3ex}}
\newcommand{\tae}[2]{#1 \models \Box_{\text{tae}} #2}
\newcommand{\mrae}[3]{#1 \to [#2]  \Box_{\text{tae}} #3}
\newcommand{\rae}[2]{[#1] \Box_{\text{tae}} #2}
\newcommand{\sae}[4]{#1 \vdash [#2] \Box_{{#4}\text{sae}} #3 }
\newcommand{\ntae}[2]{#1 \not \models \Box_{\text{tae}} #2 }

\newcommand{\cl}[1]{\overline{#1}}
\newcommand{\formulaSymbol}{\phi}
\newcommand{\secondFormulaSymbol}{\psi}
\newcommand{\FOLformulaSymbol}{P}
\newcommand{\secondFOLformulaSymbol}{Q}

\newcommand{\tracesymbol}{\sigma}
\newcommand{\firstSubtraceSymbol}{\xi}
\newcommand{\secondSubtraceSymbol}{\eta}
\newcommand{\firstStateSymbol}{\omega}
\newcommand{\secondStateSymbol}{\nu}
\newcommand{\PdTL}{PdTL\xspace}

\newcommand{\reach}{\rho}

\newcommand{\evalInState}[2]{#1(#2)}

\newcommand{\updateState}[3]{#1_{#2}^{#3}}

\newcommand{\evDomainConstraint}{R}
\newcommand{\variableSet}{\Sigma}
\newcommand{\realValuedVariableSet}{\Sigma_{\text{var}}}
\newcommand{\stateSet}[1]{\text{Sta}(#1)}
\newcommand{\termSet}[1]{\text{Trm}(#1)}
\newcommand{\setAssociTo}[1]{\llbracket #1 \rrbracket}
\newcommand{\val}[2]{val(#1, #2)}
\RequirePackage{etex}

\newif\iflongversion
\longversiontrue
\iflongversion
\usepackage[keepproof]{endproof2}
\else
\usepackage[proofatend,proofonly]{endproof2}
\fi

\usepackage{prettyref}
\newcommand{\rref}[2][]{\prettyref{#2}}
\newrefformat{sec}{Section\,\ref{#1}}
\newrefformat{def}{Def.\,\ref{#1}}
\newrefformat{thm}{Theorem\,\ref{#1}}
\newrefformat{prop}{Proposition\,\ref{#1}}
\newrefformat{lem}{Lemma\,\ref{#1}}
\newrefformat{cor}{Corollary\,\ref{#1}}
\newrefformat{ex}{Example\,\ref{#1}}
\newrefformat{tab}{Table\,\ref{#1}}
\newrefformat{fig}{Fig.\,\ref{#1}}
\newrefformat{rem}{Remark\,\ref{#1}}
\newrefformat{eqn}{\ref{#1}}
\iflongversion
\newrefformat{app}{Appendix\,\ref{#1}}
\else
\newrefformat{app}{\cite{DBLP:journals/corr/abs-some-number}}
\fi

\if@quantifiertype
\newcommand{\lquantifier}[4]
           {#1%
           \ifthenelse{\equal{#2}{}}
           {#3{\,}}
           {{#3}\,{:}\,{#2}{~}}%
           #4}
\else
\newcommand{\lquantifier}[4]
           {#1%
           #3{\,}
           #4}
\fi

\usepackage{ifpdf}
\ifpdf
\fi

\usepackage{fancyhdr}
\title{Towards Physical Hybrid Systems}
\author{Katherine Cordwell \and
Andr\'e Platzer\thanks{
Computer Science Department, Carnegie Mellon University, Pittsburgh, USA, Fakult\"at f\"ur Informatik, Technische Universit\"at M\"unchen
$\{$kcordwel,aplatzer$\}$@cs.cmu.edu
}
}
\date{}

\begin{document}
\maketitle
\allowdisplaybreaks
\thispagestyle{empty}

\begin{abstract}
Some hybrid systems models are unsafe for mathematically correct but physically unrealistic reasons. For example, mathematical models can classify a system as being unsafe on a set that is too small to have physical importance. In particular, differences in measure zero sets in models of cyber-physical systems (CPS) have significant mathematical impact on the mathematical safety of these models even though differences on measure zero sets have no tangible physical effect in a real system. We develop the concept of ``physical hybrid systems'' (PHS) to help reunite mathematical models with physical reality. We modify a hybrid systems logic (differential temporal dynamic logic) by adding a first-class operator to elide distinctions on measure zero sets of time within CPS models. This approach facilitates modeling since it admits the verification of a wider class of models, including some physically realistic models that would otherwise be classified as mathematically unsafe. We also develop a proof calculus to help with the verification of PHS.
\end{abstract}

\section{Introduction}
\textit{Hybrid systems} \cite{DBLP:conf/hybrid/NerodeK92a,DBLP:conf/hybrid/AlurCHH92}, which have interacting discrete and continuous dynamics, provide all the necessary mathematical precision to describe and verify the behavior of safety-critical \textit{cyber-physical systems} (CPS), such as self-driving cars, surgical robots, and drones.
Ironically, however, hybrid systems provide so much mathematical precision that they can distinguish models that exhibit no physically measurable difference.
More specifically, since mathematical models are minutely precise, models can classify systems as being unsafe on minutely small sets---even when these sets have no physical significance.
For example, a mathematical model could classify a system as being mathematically unsafe at a single instant in time---but why should the safety of a model give more weight to such glitches than even the very notion of solutions of differential equations, which is unaffected \cite{walter} by changes on sets of measure zero in time?
Practically speaking, a physical system is only unsafe at a single instant of time if it is also already unsafe at a significantly larger set of times.
In the worst case, such degenerate counterexamples could detract attention from real unsafeties in a model.

That is why this paper calls for a shift in perspective toward \emph{physical hybrid systems} (PHS) that are more attuned to the limitations and necessities of physics than pure mathematical models.
PHS are hybrid systems that behave safely ``almost everywhere'' (in a measure theoretic sense) and thus, physically speaking, are safe systems.
While different flavors of attaining PHS are possible and should be pursued, we propose arguably the tamest one, which merely disregards differences in safety on sets of time of measure zero.
As our ultimate hope is that models of PHS can be (correctly) formally verified without introducing any burden on the user, we introduce the ability to rigorously ignore sets of time of measure zero into logic.
A major difficulty is that there is a delicate tradeoff between the physical practicality of a definition (what real-world behavior it captures) and the logical practicality of a definition (what logical reasoning principles it supports).
Our notion of safety almost everywhere in time not only enjoys a direct link with well-established mathematical principles of differential equations, but also satisfies key logical properties, such as compositionality.

We modify \textit{differential temporal dynamic logic} (\dTL) \cite{DBLP:conf/cade/JeanninP14,Pl} to capture the notion of safety \textit{time almost everywhere} (tae) along the execution trace of a hybrid system.
\dTL extends the hybrid systems logic \textit{differential dynamic logic} (\dL) with the ability to analyze system behavior over time.
We call our new logic \textit{physical differential temporal dynamic logic} (\PdTL) to reflect its purpose.
While \PdTL is closely related to \dTL in style and development, the formalization of safety tae is entirely new, and thus requires new reasoning.
Guiding the development of \PdTL is the following motivating example: Consider a train and a safety condition $v{<}100$ on the velocity of the train. 
Physically speaking, it is fine to allow $v{=}100$ for a split-second, because this has no measurable impact. 
\PdTL is designed to classify the situation where the train continuously accelerates until $v{=}100$ and immediately brakes whenever it reaches $v{=}100$ as tae safe.

\section{Related Work} \label{sec:RelatedWork} 
Since all systems inherently suffer from imprecision, several approaches develop \textit{robust hybrid systems}, which are stable up to small perturbations---for example, in the contexts of decidability \cite{DBLP:conf/lics/AsarinB01,DBLP:conf/csl/Franzle99,DBLP:conf/lics/GaoAC12}, runtime monitoring \cite{DBLP:conf/cav/DonzeFM13,STL}, and controls \cite{MANTHANWAR20051249,mayhew1,DBLP:conf/hybrid/MoorD01}.
If systems are robust, they provide a fair amount of automation \cite{DBLP:conf/lics/GaoAC12,DBLP:conf/tacas/KongGCC15}.
Both our approach and robustness hinge on building in an awareness of physics to hybrid systems verification.
However, robustness is fundamentally different from our deductive verification approach.
The analysis of robust systems often relies on a reachability analysis (see, e.g., \cite{DBLP:journals/tcs/MoggiFDT18}) which loses much of the logical precision present in deductive verification, e.g., decidability of differential equation invariants \cite{DBLP:conf/lics/PlatzerT18}.
Further, by building on \dL, which is a general purpose hybrid systems logic, we are able to handle a wide class of models, whereas tools like dReach \cite{DBLP:conf/tacas/KongGCC15} tend to be slightly more limited in scope.
In particular, our deductive approach for PHS admits an induction principle---making it possible to verify safety properties of controllers that run in loops for any amount of time---whereas robustness approaches are presently limited to bounded model checking.
We advocate for robustness in that models can, and should be, written with an awareness of imprecision.
However, we recognize that modeling is difficult, and even well-intentioned models can suffer from nonphysical glitches---hence PHS.

Non-classical solutions of ODEs \cite{DBLP:journals/scl/Ceragioli02,Cortes2008}, especially Filippov and Carath\'{e}odory solutions, align well with the PHS intuition as they often inherently ignore sets of measure zero.
\textit{Filippov solutions} consider vector fields equivalent up to differences on sets of measure zero.
\textit{Carath\'{e}odory solutions} satisfy a differential equation everywhere except on a set of measure zero.
Hybrid systems models do not usually make use of non-classical solutions, since admitting them would require a relaxed notion of safety.
Non-classical solutions are sometimes used in the context of controls: Goebel et al.\ \cite{GOEBEL20041} generalized the notion of a solution to a hybrid system by using non-classical solutions of ODEs by Filippov and Krasovskii, with a view towards obtaining robustness properties, and a later work \cite{HDS} allows solutions that fit a system of ODEs almost everywhere in a time interval.
The temporal approach of \PdTL naturally admits Carath\'{e}odory solutions.

Eliding sets of measure zero can be computationally significant---notably, in quantifier elimination, which arises in the last step of hybrid systems proofs.
Despite having no physical meaning, measure zero sets have a significant impact on the efficiency of real arithmetic, and thus on the overall hybrid systems proofs.
The enabling factor behind efficient arithmetic \cite{DBLP:journals/cj/McCallum93} is to ignore sets of measure zero and thus remove the need to compute with irrational algebraic numbers. A potential computational benefit is an encouraging motivation for PHS.

\section{Syntax of \PdTL}\label{sec:syntax}
In pursuit of enabling the statement of physical safety properties of hybrid systems, we develop \PdTL, which builds on concepts from \dTL \cite{DBLP:conf/cade/JeanninP14,Pl} to introduce an ``almost everywhere in time'' operator, $\Box_{\text{tae}}$, which makes it possible to disregard minor glitches violating safety conditions on sets of time of measure zero.
When possible, we keep our notation consistent with that of \dTL \cite{Pl}, so that the syntax of \PdTL is very similar to the syntax of \dTL---the key difference being that we eschew \dTL trace formulas $\Box \formulaSymbol$ and $\Diamond \formulaSymbol$ in favor of $\Box_{\text{tae}} \formulaSymbol$.

The new \PdTL formula $\rae{\alpha}{\formulaSymbol}$ expresses that along each run of the hybrid system $\alpha$, the formula $\formulaSymbol$ is true at almost every time (``tae'' stands for ``time almost everywhere'').
This formula remains true even in cases where $\formulaSymbol$ is false at only a measure zero set of points in time along $\alpha$.
Because a hybrid system may exhibit different behaviors, the particular measure zero set of points in time at which $\formulaSymbol$ is false is allowed to depend on the particular run of $\alpha$.

Fix a set $\variableSet$ containing real-valued variables, function symbols, and predicate symbols.  In particular, $\variableSet$ contains the symbols needed for first-order logic of real arithmetic (FOL). We use $\realValuedVariableSet$ to denote the set of real-valued variables in $\variableSet$, and let $\termSet{\variableSet}$ denote the set of (polynomial) terms over $\variableSet$ (as in FOL).

We now define the syntax of \textit{hybrid programs} (which model hybrid systems) and formulas capable of expressing physical properties of hybrid programs.
Hybrid programs \cite{Pl,Pl2} are allowed to assign values to variables (with the $:=$ operator), test the truth of formulas (with the $?$ operator), evolve along systems of differential equations, and branch nondeterministically (with the $\cup$ operator). Hybrid programs are also sequentially composable with the $;$ operator, and can be run in loops with the $^*$ operator.

\begin{definition}\label{def:hybridProgramsGrammar} Hybrid programs are given by a grammar, where $\alpha$ and $\beta$ are hybrid programs, $e \in \termSet{\variableSet}$, $x$ is a variable, and $\FOLformulaSymbol$ and $\evDomainConstraint$ are FOL formulas:
\[\alpha, \beta ~\bebecomes~ x := e \ |\ ?\FOLformulaSymbol \ |\ x' = f(x)\& \evDomainConstraint\ |\ \alpha \cup \beta \ |\ \alpha; \beta\ |\ \alpha^* \]
\end{definition}
As in CTL$^*$ \cite{dam} and \dTL, we split \textit{\PdTL formulas} into state formulas that are true or false in a \textit{state} (i.e., at a snapshot in time) and trace formulas that are true or false along a fixed \textit{trace} that keeps track of the behavior of a system over time.

\begin{definition}\label{def:grammar} The state formulas are given by the following grammar, where $p \in \variableSet$ is a predicate symbol of arity $n{\geq}0$, $e_1, \dots, e_n \in \termSet{\variableSet}$, $\formulaSymbol$ and $\secondFormulaSymbol$ are state formulas, $\alpha$ is a hybrid program, $\kappa$ is a trace formula, and $x$ is a variable:
\[\formulaSymbol, \secondFormulaSymbol ~\bebecomes~ p(e_1, \dots, e_n) \ |\ \lnot \formulaSymbol\ |\ \formulaSymbol \land \secondFormulaSymbol\ |\ \forall x\, \formulaSymbol\ |\ [\alpha]\kappa\ |\ \langle \alpha \rangle \kappa \]
Trace formulas are given by the following grammar, where $\formulaSymbol$ is a state formula:
\[\kappa ~\bebecomes~ \formulaSymbol\ |\ \Box_{\text{tae}} \formulaSymbol\]
\end{definition}
We will also allow the use of the standard logical operators $\lor$, $\to$, and $\leftrightarrow$, which are defined in terms of $\lnot$ and $\land$ as usual in classical logic.
 
Our motivating example can be modeled in \PdTL as follows:
\begin{align*} a{=}0 \land v{=}0 \to  & [\big(((?(v{<}100); a := 1) \cup (?(v{=}100); a := -1)); \\
 & \hspace{2em} \{x'{=}v, v'{=}a \ \&\ 0{\leq}v{\leq}100\}\big)^*] \Box_{\text{tae}}\, v{<}100
\end{align*}
This claims that if the initial velocity and acceleration are both 0, then along any run of the system, $v{<}100$ holds at almost all times.
The train accelerates if $v{<}100$ and brakes if $v{=}100$; it moves according to the system of ODEs $x'{=}v, v'{=}a$.
The evolution domain constraint $v{\leq}100$ indicates an event-triggered controller \cite{Pl2}.

A natural question is why we choose to build in reasoning about $\Box_{\text{tae}}$ by developing \PdTL instead of having the user edit the model by, for example, making the postcondition $v{\leq}100$ instead of $v{<}100$.
Indeed, in this particular case that would make the program safe at every moment in time.
However, in other examples, editing the postcondition in a similar way may be unwise.
For example, although $\rae{x := 0; y := 0; \{x' = 0, y' = 1\}}(y{>}0 \to x{>}1\lor x{<}1)$ is valid, if we relax the inequalities, $\rae{x := 0; y := 0; \{x' = 0, y' = 1\}}(y{\geq}0 \to x{\geq}1\lor x{\leq}1)$ is not.
As reasoning about hybrid systems is so subtle, the most user-friendly approach to eliding sets of measure zero is to specifically introduce the rigorous ability to ignore sets of measure zero into the logic.

\section{Semantics of \PdTL}\label{sec:Semantics}
We now report a trace semantics for hybrid programs, based on which we give meaning to the informal concept of formulas being true almost everywhere in time, first along an individual trace and then along all traces of a hybrid program.

\subsection{Semantics of State Formulas}
State formulas are evaluated at states, which capture the behavior of the hybrid program at an instant in time. Each state contains the values of all relevant variables at a given instant.
We formalize this in the following definition.
\begin{definition} A \textit{state} is a map $\firstStateSymbol: \realValuedVariableSet \to \mathbb{R}$.
We distinguish a separate state $\Lambda$ to indicate the failure of a system run.
The set of all states is $\stateSet{\realValuedVariableSet}$. 
\end{definition}

We now give the semantics of state formulas.
The $\val{\firstStateSymbol}{\formulaSymbol}$ operator determines whether state formula $\formulaSymbol$ is true or false in state $\firstStateSymbol$.
The valuations of the state formulas $[\alpha] \kappa$ and $\langle \alpha \rangle \kappa$ depend on the semantics of traces $\tracesymbol$ (especially the notion of $\text{first }\tracesymbol$), which is explained in \rref{def:dTLTraceDef}, and on the semantics of the trace formula $\kappa$ (i.e., $\val{\firstStateSymbol}{\kappa}$) which is given later, in \rref{def:traceformulasemantics}, and may be undefined.
\begin{definition}[\cite{Pl}]\label{def:stateformulaSemantics} 
The valuation of state formulas with respect to state $\firstStateSymbol$ is defined inductively:
\begin{enumerate}
\item $\val{\firstStateSymbol}{p(\theta_1, \dots, \theta_n)}$ is $p^{\ell}(\val{\firstStateSymbol}{\theta_1}, \dots, \val{\firstStateSymbol}{\theta_n})$ where $p^{\ell}$ is the relation associated with $p$ under the semantics of real arithmetic
\item $\val{\firstStateSymbol}{\lnot \formulaSymbol}$ is true iff $\val{\firstStateSymbol}{\formulaSymbol}$ is false
\item $\val{\firstStateSymbol}{\formulaSymbol \land \secondFormulaSymbol} $ is true  iff $\val{\firstStateSymbol}{\formulaSymbol}$ is true and $\val{\firstStateSymbol}{\secondFormulaSymbol}$ is true
\item $\val{\firstStateSymbol}{\forall x\, \formulaSymbol} $ is true  iff $\val{\updateState{\firstStateSymbol}{x}{d}}{\formulaSymbol} $ is true  for all $d \in \mathbb{R}$, where $\updateState{\firstStateSymbol}{x}{d}$ is the state that is identical to $\firstStateSymbol$, except $x$ has the value $d$.
\item $\val{\firstStateSymbol}{[\alpha] \kappa}$ is true  iff for every trace $\tracesymbol$ of $\alpha$ that starts in $\text{first }\tracesymbol = \firstStateSymbol$, if $\val{\tracesymbol}{\kappa}$ is defined, then $\val{\tracesymbol}{\kappa}$ is true 
\item $\val{\firstStateSymbol}{\langle \alpha \rangle \kappa} $ is true  iff there is some trace $\tracesymbol$ of $\alpha$ where $\text{first }\tracesymbol = \firstStateSymbol$ and $\val{\tracesymbol}{\kappa}$ is true
\end{enumerate}
We write $\firstStateSymbol \models \formulaSymbol$ when $\val{\firstStateSymbol}{\formulaSymbol}$ is true.
We write $\firstStateSymbol \not\models \formulaSymbol$ when $\val{\firstStateSymbol}{\formulaSymbol}$ is false.
\end{definition}

\subsection{Traces of Hybrid Programs}\label{sec:TracesSection} 
Trace formulas are evaluated with respect to an execution trace of a hybrid program.
Intuitively, a trace of a hybrid program views its behavior over time as a sequence of functions, where each function corresponds to a particular discrete or continuous portion of the dynamics.
Most hybrid systems are associated with multiple traces (to reflect the variety of behaviors that a given program can exhibit).
Each function within a trace maps from a time interval to states of the hybrid program.  
Continuous portions of traces are functions from an uncountable time interval, and thus are associated to uncountably many states, whereas discrete portions involve just a single state.
Significantly, traces are allowed to end in an abort state $\Lambda$, which indicates an unsuccessful run of a program.
Aborts are incurred when tests fail and when evolution domain constraints are not initially satisfied.
No program can run past $\Lambda$.
We review the formal definition below.

\begin{definition}[{\cite{Pl}}]\label{def:dTLTraceDef}
A \textit{trace} $\tracesymbol$ of a hybrid program $\alpha$ is a sequence of functions $\tracesymbol = (\tracesymbol_0, \tracesymbol_1, \dots, \tracesymbol_n)$ where $\tracesymbol_i: [0, r_i] \to \stateSet{\realValuedVariableSet}$.
We will denote the length of the interval associated to $\tracesymbol_i$ by $|\tracesymbol_i|$ (so that if $\tracesymbol_i$ maps from $[0, r_i]$ to states of $\alpha$, $|\tracesymbol_i| = r_i$).
A \textit{position} of $\tracesymbol$ is a tuple $(i, \zeta)$ where $i \in \mathbb{N}$, $\zeta \in [0, r_i]$.
Each position $(i, \zeta)$ is associated with the corresponding state $\tracesymbol_i(\zeta)$.
A trace $(\tracesymbol_0, \tracesymbol_1, \dots, \tracesymbol_n)$ is said to \textit{terminate} if it does not end in the abort state, i.e. if $\tracesymbol_n(|\tracesymbol_n|) \neq \Lambda$,
and we write $\text{last } \tracesymbol \equiv \tracesymbol_n(|\tracesymbol_n|)$ in that case.
We write $\text{first } \tracesymbol \equiv \tracesymbol_0(0)$ for the first state.
\end{definition}

We now modify the trace semantics \cite{Pl} using Carath\'{e}odory solutions for ODEs \cite{walter} using notation as in \cite[Definition 2.6]{Pl2}.

\begin{definition}\label{def:CaraDef} The state $\secondStateSymbol$ is \emph{reachable in the extended sense} from initial state $\firstStateSymbol$ by $x_1' = \theta_1, \dots, x_n' = \theta_n\ \&\ \evDomainConstraint$ iff there is a function $\varphi : [0,r] \to \stateSet{\realValuedVariableSet}$ s.t.:
\begin{enumerate}
\item Initial and final states match: $\varphi(0) = \firstStateSymbol, \varphi(r) = \secondStateSymbol$.
\item $\varphi$ is absolutely continuous.
\item $\varphi$ respects the differential equations almost everywhere: For each variable $x_i$, $\evalInState{\varphi(z)}{x_i}$ is continuous in $z$ on $[0, r]$ and if $r {>} 0$, $\evalInState{\varphi(z)}{x_i}$ has a time-derivative of value $\evalInState{\varphi(z)}{\theta_i}$ at all $z \in [0, r]\setminus \mathcal{U}$, for some set $\mathcal{U} \subset [0, r]$ that has Lebesgue measure zero.
\item The value of other variables $y \not \in \{x_1, \dots, x_n\}$ remains constant throughout the continuous evolution, that is $\evalInState{\varphi(z)}{y} = \evalInState{\firstStateSymbol}{y}$ for all times $z \in [0,r]$;
\item $\varphi$ respects the evolution domain at all times: $\varphi(z) \models \evDomainConstraint$ for all $z \in [0, r]$.
\end{enumerate}
If such a $\varphi$ exists, we say that $\varphi \models x_1' = \theta_1 \land \cdots \land x_n' = \theta_n\ \&\ \evDomainConstraint$ almost everywhere.
\end{definition}
This change highlights how the PHS intuition aligns with the intuition behind Carath\'{e}odory solutions.
However, we have left the syntax of hybrid programs unchanged, and any system of differential equations in this syntax has a unique classical solution by Picard-Lindel\"{o}f \cite{Pl2}.
We believe that in order to determine a suitable generalization of the syntax of hybrid programs to allow systems of ODEs with Carath\'{e}odory solutions, one should first develop strategies for reasoning about ODEs in \PdTL that are beyond the scope of this work (for example, a notion of differential invariants---see \cite{Pl2,DBLP:conf/lics/PlatzerT18}).

Note that in condition 5 of \rref{def:CaraDef}, the solution is required to stay within the evolution domain constraint at \emph{all} times. 
This is because evolution domain constraints, when used correctly, are nonnegotiable---for example, a correct use of evolution domain constraints is to reflect some underlying property of physics, like that the speed of a decelerating system is always nonnegative.

We define the trace semantics \cite{Pl}, with the above change for ODEs.
Like evolution domain constraints, we treat tests as nonnegotiable, so as not to interfere with a user's ability to write precise models.
We do not intend to secretly change the meaning of models, but rather to identify physically correct models.
\begin{definition} \label{def:dTLTraceSemantics}  The \textit{trace semantics}, $\tau(\alpha)$, of a hybrid program $\alpha$ is the set of all its possible hybrid traces and is defined inductively as follows (where $e \in \termSet{\variableSet}$, $x' = f(x)$ is a vectorial ODE, $\evDomainConstraint$ and $\FOLformulaSymbol$ are FOL formulas, $\beta$ is a hybrid program, and where for a state $\firstStateSymbol$, $\hat{\firstStateSymbol}$ is the function from $[0, 0] \to \stateSet{\realValuedVariableSet}$ with $\hat{\firstStateSymbol}(0) = \firstStateSymbol)$:
\begin{enumerate}
\item $\tau(x:=e) = \{(\hat{\firstStateSymbol}, \hat{\secondStateSymbol}): \secondStateSymbol = \updateState{\firstStateSymbol}{x}{val(\firstStateSymbol, e)} \text{ for } \firstStateSymbol \in \stateSet{\realValuedVariableSet}\}$
\item $\tau(x'{=}f(x)\ \&\ \evDomainConstraint) = \{(\varphi): \varphi \models x'{=}f(x)\ \&\ \evDomainConstraint \text{ almost everywhere } \} \cup \{(\hat{\firstStateSymbol}, \hat{\Lambda}): \firstStateSymbol \not \models \evDomainConstraint \}$
\item $\tau(\alpha \cup \beta) = \tau(\alpha) \cup \tau(\beta)$
\item $\tau(?\FOLformulaSymbol) = \{(\hat{\firstStateSymbol}): val(\firstStateSymbol, \FOLformulaSymbol) = true\} \cup \{(\hat{\firstStateSymbol}, \hat{\Lambda}) : val(\firstStateSymbol, \FOLformulaSymbol) = false\}$
\item $\tau(\alpha; \beta) = \{\tracesymbol \circ \zeta: \tracesymbol \in \tau(\alpha), \zeta \in \tau(\beta) \text{ when } \tracesymbol \circ \zeta \text{ is defined}\}$; the composition of $\tracesymbol = (\tracesymbol_0, \tracesymbol_1, \tracesymbol_2, \dots)$ and $\zeta = (\zeta_0, \zeta_1, \zeta_2, \dots)$ is 
\[ \tracesymbol~\circ~\zeta := 
\begin{cases} 
(\tracesymbol_0, \dots, \tracesymbol_n, \zeta_0, \zeta_1, \dots) &\text{ if }\tracesymbol \text{ terminates at }\tracesymbol_n \text{ and last }\tracesymbol{=}\text{first }\zeta \\
\tracesymbol &\text{ if }\tracesymbol \text{ does not terminate}\\
\text{not defined} & \text{otherwise}
\end{cases}
\]
\item $\tau(\alpha^*) = \bigcup_{n \in \mathbb{N}}\tau(\alpha^n)$, where $\alpha^{n+1} := (\alpha^n; \alpha)$ for $n{\geq}1$, and $\alpha^0 :=\ ?(true)$
\end{enumerate}

\end{definition}
As an important remark, notice that if we have a trace $\tracesymbol = (\tracesymbol_0, \dots, \tracesymbol_n)$ and $|\tracesymbol_i| {>} 0$, then $\tracesymbol_i \in \tau(x'=f(x)\ \&\ \evDomainConstraint)$ for some (vectorial) ODE $x'=f(x)$ and some evolution domain constraint $\evDomainConstraint$. In other words, only continuous portions of a trace have nonzero duration, and continuous portions are only introduced when our system is evolving subject to a system of differential equations.

We are almost ready to give the semantics of trace formulas, but first we need to take a slight detour to discuss what formulas make ``physical sense''.

\subsection{Physical Formulas} \label{sec:PhysicalFormulasSection}
One feature of our motivating example is that $v{<}100$ is not a physically meaningful postcondition: If $v$ is allowed to get arbitrarily close to $100$, $v{=}100$ should also be allowed, since there is no physically measurable difference between $v{<}100$ and $v{\leq}100$. The postcondition $v{\leq}100$ is, mathematically speaking, less restrictive than $v{<}100$, but also practically speaking, the same as $v{<}100$.  Motivated by this intuition, we define ``physical formulas''.

\subsubsection{Physical Formulas and $\cl{\formulaSymbol}$}
Geometrically, the set of states in which a state formula $\formulaSymbol$ in $n$ variables is true is a subset, $\setAssociTo{\formulaSymbol} = \{(x_1, \dots, x_n) \in \mathbb{R}^n\ |\ \formulaSymbol(x_1, \dots, x_n)\}$, of $\mathbb{R}^n$. We use this correspondence to define the physical version of $\formulaSymbol$.

\begin{definition}\label{def:physicalPSemantic} A state formula  $\formulaSymbol$ is called \emph{physical} iff $\setAssociTo{\formulaSymbol}$ is topologically closed. If $\formulaSymbol$ is a formula in $n$ variables $x_1, \dots, x_n$, then the \textit{physical version} of $\formulaSymbol$ is the closure, denoted by $\cl{\formulaSymbol}$, and is, indeed, definable \cite{BCR} by:
\[\forall \epsilon {>} 0\, \exists y_1, \dots, y_n \ \left(\formulaSymbol(y_1, \dots, y_n) \land (x_1 - y_1)^2 + \cdots + (x_n - y_n)^2{<}\epsilon^2\right)
\]
This satisfies $\setAssociTo{\cl{\formulaSymbol}} = \overline{\setAssociTo{\formulaSymbol}}$, where $\overline{\setAssociTo{\formulaSymbol}}$ is the topological closure of $\setAssociTo{\formulaSymbol}$.
Quantifier elimination can compute a quantifier-free equivalent of $\cl{\formulaSymbol}$ that is often preferable.
\end{definition}

Associating a state formula to a subset of $\mathbb{R}^n$ is useful for identifying which points are ``almost included'', which will be crucial knowledge for our temporal approach. In the train example, $v{=}100$ is ``almost included'' in the postcondition $v{<}100$. These points that are ``almost included'' in formula $\formulaSymbol$ are exactly the limit points of the set associated to $\formulaSymbol$, so $\overline{\setAssociTo{\formulaSymbol}}$ adds in all of these limit points.
We will make use of the following properties of the physical version of $\formulaSymbol$.

\begin{proposition}\label{prop:clProp2} For any state formula $\formulaSymbol$, $\formulaSymbol \to \cl{\formulaSymbol}$ is valid (i.e., true in all states).
\end{proposition}

\begin{proofatend}
By properties of the topological closure of a set, $\setAssociTo{\formulaSymbol} \subseteq \overline{\setAssociTo{\formulaSymbol}}$.
Since $\overline{\setAssociTo{\formulaSymbol}} = \setAssociTo{\cl{\formulaSymbol}}$ from \rref{def:physicalPSemantic}, we have  $\setAssociTo{\formulaSymbol} \subseteq \setAssociTo{\cl{\formulaSymbol}}$,  and so given any state $\firstStateSymbol$, if $\firstStateSymbol \models \formulaSymbol$ is true, then $\firstStateSymbol \models \cl{\formulaSymbol}$ is true.
Thus the formula $\formulaSymbol \to \cl{\formulaSymbol}$ is valid. 
\end{proofatend}

\begin{proposition}\label{prop:clProp1}
The following proof rule is sound for state formulas $\formulaSymbol, \secondFormulaSymbol$ (i.e., the validity of all premises implies the validity of the conclusion):
\begin{center}
\begin{prooftree}
\hypo{\formulaSymbol \to \secondFormulaSymbol}
\infer1[TopCl]{\cl{\formulaSymbol} \to \cl{\secondFormulaSymbol}}
\end{prooftree}
\end{center}
\end{proposition}

\begin{proofatend}
Assume $\formulaSymbol \to \secondFormulaSymbol$ is valid.  Then every state $\secondStateSymbol$ with $\secondStateSymbol \models \formulaSymbol$ also satisfies $\secondStateSymbol \models \secondFormulaSymbol$.  Equivalently, $\setAssociTo{\formulaSymbol} \subseteq \setAssociTo{\secondFormulaSymbol}$.  By properties of the closure of a set, we then also have $\overline{\setAssociTo{\formulaSymbol}} \subseteq \overline{\setAssociTo{\secondFormulaSymbol}}$. This is equivalent to $\setAssociTo{\cl{\formulaSymbol}} \subseteq \setAssociTo{\cl{\secondFormulaSymbol}}$ by \rref{def:physicalPSemantic}, and so any state $\secondStateSymbol$ with $\secondStateSymbol \models \cl{\formulaSymbol}$ also satisfies $\secondStateSymbol \models \cl{\secondFormulaSymbol}.$  Thus, $\cl{\formulaSymbol} \to \cl{\secondFormulaSymbol}$ is valid.
\end{proofatend}

\subsection{Semantics of Trace Formulas}\label{sec:LogicalConstructsSection}
Intuitively, we want to say that, for a trace $\tracesymbol$, $\tracesymbol \models \Box_{\text{t.a.e}} \formulaSymbol$ when there is only a ``small'' set of positions $(i, \zeta)$ where $\tracesymbol_i(\zeta) \not \models \formulaSymbol$ and where the discrete portions of $\tracesymbol$ satisfy a reasonable constraint.
To formalize the notion of a ``small'' set of positions, we map positions of $\tracesymbol$ to $\mathbb{R}$, since $\mathbb{R}$ admits the Lebesgue measure.

\begin{definition}\label{def:taef}
Given a trace $\tracesymbol = (\tracesymbol_0, \dots, \tracesymbol_n)$ of a hybrid program $\alpha$ with $|\tracesymbol_i| = r_i$, map each position $(i, \zeta)$ of $\tracesymbol$ to $\zeta + i + \sum_{k=0}^{i-1} |\tracesymbol_k|$, so that the positions $(0, 0), \dots, (0, r_0)$ cover the interval $[0, r_0]$, the positions $(1, 0), \dots, (1, r_1)$ cover the interval $[r_0 + 1, r_0 + r_1 + 1]$, and so on.
In this way we have an injection, which we call $f$, from positions of a trace $\tracesymbol$ to (a subset of) $\mathbb{R}$.
\end{definition}

\rref{fig:taemap} illustrates the mapping $f$, which is obtained by first concatenating the positions between each discrete step (i.e. the positions for each continuous function $\tracesymbol_i$) and then projecting these concatenations onto a single time axis so that the images of the states for $\tracesymbol_i$ and the states for $\tracesymbol_j$ are disjoint when $i \neq j$. To ensure disjointness, our mapping places an open interval of unit length between the images of the states of $\tracesymbol_i$ and the states of $\tracesymbol_{i+1}$ (for all $i$). The unit length was chosen arbitrarily---any nonzero length would work---but it is important that the positions $(i, r_i)$ and $(i+1, 0)$ have different projections, since discrete changes can cause their states $\tracesymbol_i(r_i)$ and $\tracesymbol_{i+1}(0)$ to be different.
\begin{figure}[!ht]
\begin{center}
\includegraphics[scale = .5]{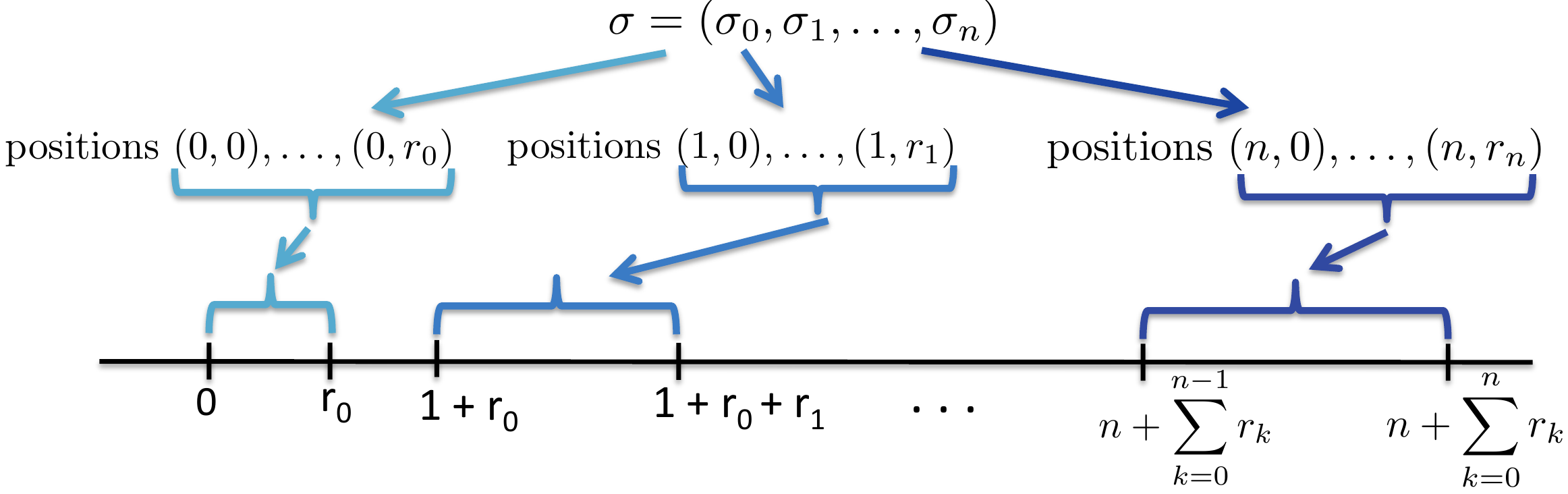}
\caption{Injectively mapping positions of a trace to times in $\mathbb{R}$}
\label{fig:taemap}
\end{center}
\end{figure}

\begin{definition}\label{def:traceformulasemantics} 
The valuation of a trace formula $\kappa$ with respect to trace $\tracesymbol$ is defined as:
\begin{enumerate}
    \item $\val{\tracesymbol}{\formulaSymbol} = \val{\text{last }\tracesymbol}{\formulaSymbol}$ if $\tracesymbol$ terminates.
    If $\tracesymbol$ does not terminate, then $\val{\tracesymbol}{\formulaSymbol}$ is undefined. 
    We write $\tracesymbol \models \formulaSymbol$ when $\val{\tracesymbol}{\formulaSymbol}$ is true.
    We write $\tracesymbol \not \models \formulaSymbol$ when $\val{\tracesymbol}{\formulaSymbol}$ is false.
    \item Let $\mathcal{U}$ be the set of positions $(i, \zeta)$ where corresponding states $\tracesymbol_i(\zeta)$ satisfy $\tracesymbol_i(\zeta) \not \models \formulaSymbol$ and $\tracesymbol_i(\zeta) \neq \Lambda$.
    We say that $\val{\tracesymbol}{\Box_{\text{tae}}\formulaSymbol}$ is true iff the following two conditions are satisfied:
\begin{enumerate}
\item (Discrete condition) For all $i$, if $|\tracesymbol_i| = 0$ and $\tracesymbol_i(0) \neq \Lambda$, then $\tracesymbol_i(0) \models \cl{\formulaSymbol}$.
\item (Continuous condition) $f(\mathcal{U}) \subseteq \mathbb{R}$ has measure zero with respect to the Lebesgue measure, where $f$ is the mapping defined in \rref{def:taef} for $\tracesymbol$; i.e., for all $\epsilon {>} 0$ there exist intervals $I_p = [a_p, b_p]$ so that $f(\mathcal{U}) \subseteq \bigcup_{p = 0}^{\infty} I_p$ and $\sum_{p = 0}^{\infty} |b_p - a_p|{<}\epsilon$ (see \cite[Section 1-1013]{M0Jeffreys}, or \cite{Royden}).
\end{enumerate}
We write $\tae{\tracesymbol}\formulaSymbol$ when $\val{\tracesymbol}{\Box_{\text{tae}}\formulaSymbol} $ is true.
We write $\ntae{\tracesymbol}\formulaSymbol$ when $\val{\tracesymbol}{\Box_{\text{tae}}\formulaSymbol}$ is false.
\end{enumerate}
\end{definition}
\noindent A short primer on the measure theory we need is in \rref{app:MeasureTheory}.

With this mapping, the condition that $f(\mathcal{U})$ has measure zero with respect to the Lebesgue measure enforces the t.a.e. constraint for the continuous portions of our program. The image of the states for $\tracesymbol_i$ is an interval of length $r_i$. In order for the t.a.e constraint to be satisfied, $\tracesymbol_i$ cannot go wrong except at a ``small set'' of states, where we use our mapping to formalize the notion of a ``small set'' in terms of measure zero.

Next, the condition that $\tracesymbol_i(0) \models \cl{\formulaSymbol}$ whenever $|\tracesymbol_i| = 0$ constrains the discrete portions of a program. This constraint will be important for induction; it also ensures that discrete programs behave reasonably within our logic. For example, let $\alpha$ be the fully discrete program $x := 5; (x := x + 1)^*$, and take a trace $\tracesymbol$ of $\alpha$. The states of $\tracesymbol$ map to the points $5, 6, 7, \dots, n$, and \textit{without the discrete condition}, we would be able to show that $\rae{x := 5; (x := x + 1)^*}{x{<}5}$ is valid, even though $x$ is never less than 5 along the trace.

To understand why the discrete condition specifies $\tracesymbol_i(0) \models \cl{\formulaSymbol}$ when $|\tracesymbol_i| = 0$ instead of $\tracesymbol_i(0) \models \formulaSymbol$ when $|\tracesymbol_i| = 0$, recall the motivating train control example.
We want to allow the velocity of the train to evolve from a safe state where $v{<}100$ to an unsafe state where $v = 100$, as long as the train then immediately brakes (sets its acceleration to a negative value).
If we specified that $\tracesymbol_i(0) \models \formulaSymbol$, the train would \textit{not} be allowed to accelerate from $v{<}100$ to $v = 100$ and then brake, because at the discrete braking point, the train would be in a state that is unsafe mathematically (though still safe physically).

Given a hybrid program $\alpha$, postcondition $\formulaSymbol$, and a state $\firstStateSymbol$, we are interested in determining whether $\firstStateSymbol \models \rae{\alpha}{\formulaSymbol}$ holds, because when such a formula is true for a given hybrid program, that indicates that no matter how that particular hybrid program runs, it will be ``safe almost everywhere''. Following \rref{def:stateformulaSemantics}, this is true iff for each trace $\tracesymbol \in \tau(\alpha)$ with $\text{first }\tracesymbol = \firstStateSymbol$, $\tae{\tracesymbol}{\formulaSymbol}$.

\section{Discussion} \label{sec:discussion}
Now that we have developed the semantics, we step back to consider how \PdTL makes progress towards PHS, and why \PdTL is a good way to introduce PHS.

\paragraph{Impact on Modeling.}
\PdTL allows the verification of several classes of physically realistic models that are mathematically not quite safe.
For example, in \rref{sec:Example} we will explain how \PdTL allows the verification of the train control model.
This simple train example is representative of a greater class of examples---tiny glitches are common in event-triggered controllers, since the event that is being detected is often an almost unsafe event that requires the controller to immediately change behavior.

Other examples do not involve time-triggered controllers, but rather suffer from tiny glitches at handover points between the discrete and continuous dynamics within a hybrid program.  Consider the safety postcondition $x^2 + y^2{<}1$ and the hybrid program $x := 0; y := 1; \{x' = -x, y' = -y\}$.
The only glitch is that the program starts ever so slightly outside the safe set.
Since it immediately moves into the safe set, all runs of this hybrid program are safe tae.
\PdTL is designed to classify $\rae{x := 0; y := 1; \{x' = -x, y' = -y\}}{x^2 + y^2{<}1}$ as valid.

In another class of examples, our approach handles tiny glitches within the continuous portion of the program.
This can easily happen if the postcondition is missing some small regions.
For example, consider two robots that are moving, one in front of the other.
Since we do not want the robots to collide, it is unsafe for the second robot to accelerate while the first robot is braking.
Say we model this with safety postcondition $\lnot(a_1{\leq}0 \land a_2{\geq}0)$.
This is a small modeling mistake, because we should allow the point where $a_1{=}0$ and $a_2{=}0$.
Now, if our controller is $a_1 := -1; a_2 := -1; \{a_1' = 1, a_2' = 1\}$, any run of this hybrid program is tae safe, but not safe at all points in time (as some runs will contain the origin).
Notably, the very similar controller $a_1 := -1; a_2 := -1; \{a_1' = 1, a_2' = 2\}$ is \textit{not} tae safe.
\PdTL is designed to distinguish between these two controllers.

\paragraph{Why tae?}\label{sec:sae}
The tae safety notion along the trace of a hybrid system is a natural approach with strong mathematical underpinnings (e.g., from the invariance of Lebesgue integrals up to sets of measure zero, and from Carath\'{e}odory solutions), and with physical motivation from examples like those just discussed.
Introducing tae is a good way to begin PHS, because it may be the closest possible PHS construct to the canonical notion of ``safety everywhere''.
However, this closeness to safety everywhere does make tae more restrictive than some other possible PHS notions---for example, a notion of safety ``space almost everywhere'', or sae.

Consider a self-driving car moving in $\mathbb{R}^3$.
The final states of a hybrid program $\alpha$ modeling the car correspond to positions in $\mathbb{R}^3$, so here we may wish to consider safety almost everywhere with respect to the Lebesgue measure \textit{on the set of all possible final states of $\alpha$} as follows: Given a hybrid program $\alpha$ and precondition $\secondFormulaSymbol$, let $\mathcal{F} = \{\text{last } \tracesymbol \text{ s.t. } \tracesymbol \in \tau(\alpha), \text{first } \tracesymbol \models \secondFormulaSymbol\}$.
Say that $\sae{\secondFormulaSymbol}{\alpha}{\formulaSymbol}{\mu}$ if $\mu(\{\firstStateSymbol \in \mathcal{F}\ |\ \firstStateSymbol \not \models \formulaSymbol\}) = 0$, where $\mu$ is the Lebesgue measure on $\mathbb{R}^3$. 
This modality has an intuitive geometric interpretation and is more permissive than tae.

However, $\Box_{\mu\text{sae}}$ applies only when there is a natural measure $\mu$ on the set of all possible final states of $\alpha$.
Furthermore, $\Box_{\mu\text{sae}}$ is not compositional.
Given $\sae{\FOLformulaSymbol}{\alpha}{\FOLformulaSymbol}{\mu}$ and  $\sae{\FOLformulaSymbol}{\beta}{\FOLformulaSymbol}{\mu}$, in order to conclude that  $\sae{\FOLformulaSymbol}{\alpha; \beta}{\FOLformulaSymbol}{\mu}$, one needs to know that $\beta$ is sae safe when starting in $\setAssociTo{\FOLformulaSymbol} \cup \secondFOLformulaSymbol$, for any measure zero set $\secondFOLformulaSymbol$ reachable from $\alpha$.
This is not always true---for example, let $P$ be $x^2 + y^2{<}1$, $\alpha$ be $\{x' = 1, y' = 1\ \&\ x^2 + y^2{\leq}1\}$ and $\beta$ be $?(x^2 + y^2{=}1); \{x' = 1, y' = 1\}$.
Further, given $\alpha$, it is unclear how to syntactically classify such sets $\secondFOLformulaSymbol$---as sae is more relaxed than tae, it also seems to be less well-behaved.
Thus, although $\Box_{\mu\text{sae}}$ has some advantages, it is not clear how to constrain it to achieve desirable logical properties like compositionality.
Although we hope that future work will develop a notion of safety sae, the challenges therein are no small matter.
In contrast, tae satisfies many nice logical properties, which we now turn our attention to.

\section{Proof Calculus and Properties of \PdTL}\label{sec:properties}
Before developing the proof calculus for \PdTL, we discuss some key properties.
First, \PdTL is a \textit{conservative} extension of \dL, i.e. all valid formulas of \dL are still valid in \PdTL.
The proof of this, discussed in \rref{app:Conservativity}, is essentially the same as the proof that \dTL is a conservative extension of \dL \cite[Proposition 4.1]{Pl}.
This conservativity property is useful, since if we are able to reduce temporal \PdTL formulas to \dL formulas, we can use the extensive machinery built for \dL to close proofs.
Indeed, our proof calculus is designed to reduce temporal \PdTL formulas into nontemporal formulas to rely on \dL's capabilities for the latter.

Key \dL axioms are proved sound for \PdTL in \rref{app:Conservativity}, which is useful as sometimes a \dL axiom is needed to reduce the goal in a proof, as we will see when we prove the train example in \rref{sec:Example}.

Next, we state three properties of temporal \PdTL formulas which underlie some of the soundness proofs for rules in the proof calculus.
These properties hold by construction.

\begin{lemma} \label{lem:laststatelemma} If $\tae{\tracesymbol}{\formulaSymbol}$ for a terminating $\tracesymbol$, then $\text{last } \tracesymbol \models \cl{\formulaSymbol}$.
\end{lemma}

\begin{proofatend}
Say that $\tracesymbol = (\tracesymbol_0, \dots, \tracesymbol_n)$ and $\tracesymbol$ terminates.
Then if $|\tracesymbol_n| = 0$, $\tracesymbol_n(0) \models \cl{\formulaSymbol}$ by \rref{def:traceformulasemantics}, which implies the result since $\tracesymbol_n(0) = \text{last }\tracesymbol$.
Else, $r_n = |\tracesymbol_n| \ >\ 0$, then using \rref{def:dTLTraceSemantics}, we see that $\tracesymbol_n \in \tau(x'{=}f(x)\ \&\ \evDomainConstraint)$, i.e. $\tracesymbol_n \models x'{=}f(x)\ \&\ \evDomainConstraint \text{ almost everywhere}$, where $x'{=}f(x)$ is a vectorial ODE and $\evDomainConstraint$ is an evolution domain constraint.

Now, assume for the sake of contradiction that $\text{last } \tracesymbol = \tracesymbol_n(r_n) \not \models \cl{\formulaSymbol}$.
Thus $\tracesymbol_n(r_n) \not \in \setAssociTo{\cl{\formulaSymbol}}$.
Since $\setAssociTo{\cl{\formulaSymbol}} = \overline{\setAssociTo{\formulaSymbol}}$ by \rref{def:physicalPSemantic}, this means that $\tracesymbol_n(r_n) \not \in \overline{\setAssociTo{\formulaSymbol}}$.
Since $\overline{\setAssociTo{\formulaSymbol}}$ is closed, there is some open set $\mathcal{U}$ with $\tracesymbol_n(r_n) \in \mathcal{U}$ and $\mathcal{U} \cap \overline{\setAssociTo{\formulaSymbol}} = \emptyset$.
Recall that Carath\'{e}odory solutions are absolutely continuous. From this, $\sigma_n$ (which is the restriction of the Carath\'{e}odory solution to $x'{=}f(x)$ from time $0$ to time $r_n$) is absolutely continuous on $[0, r_n]$.
By the definition of continuity, for every open set $\mathcal{A}$ containing $\sigma_n(r_n)$, there is an open set $\mathcal{B}$ (with respect to the subspace topology of $[0, r_n]$) with $r_n \in \mathcal{B}$ so that for all $t \in \mathcal{B}$, $\sigma_n(t) \in \mathcal{A}$.
Setting $\mathcal{A} = \mathcal{U}$, we obtain in particular that there is a half-open interval $(t, r_n]$ with $0{\leq}t{<}r_n$ so that $\tracesymbol_n(\zeta) \not \in \overline{\setAssociTo{\formulaSymbol}}$ for all $t{<}\zeta{\leq}r_n$.
From this, by \rref{prop:clProp2}, all $t{<}\zeta{\leq}r_n$ satisfy $\tracesymbol_n(\zeta) \not \models \formulaSymbol$, which contradicts $\tae{\tracesymbol}{\formulaSymbol}$ (because $(t, r_n]$ is an interval of nonzero measure on which $\formulaSymbol$ is false).
Thus $\text{last }\tracesymbol \models \cl{\formulaSymbol}$, as desired.
 \end{proofatend}

\begin{corollary}\label{cor:EM} The formula $\rae{\alpha}{\formulaSymbol} \to [\alpha]\cl{\formulaSymbol}$ is valid.
\end{corollary}
\begin{proofatend} Take some state $\firstStateSymbol$ with $\firstStateSymbol \models \rae{\alpha}{\formulaSymbol}$. We need to show that $\val{\firstStateSymbol}{[\alpha]\cl{\formulaSymbol}}$ is true.
By \rref{def:stateformulaSemantics}, $\val{\firstStateSymbol}{[\alpha]\cl{\formulaSymbol}}$ is true iff for every trace $\tracesymbol$ of $\alpha$ with first $\tracesymbol = \firstStateSymbol$, if $\val{\tracesymbol}{\cl{\formulaSymbol}}$ is defined, then $\val{\tracesymbol}{\cl{\formulaSymbol}}$ is $true$. Also by \rref{def:stateformulaSemantics}, $\val{\tracesymbol}{\cl{\formulaSymbol}}$ is $\val{\text{last }\tracesymbol}{\formulaSymbol}$ if $\tracesymbol$ terminates (and undefined else).
Now, from \rref{lem:laststatelemma}, if $\tracesymbol$ terminates, then $\val{\tracesymbol}{\cl{\formulaSymbol}}$ is true.
Otherwise there is nothing to show.
Thus $\val{\firstStateSymbol}{[\alpha]\cl{\formulaSymbol}}$ is true, as desired.
\end{proofatend}

\begin{lemma} \label{lem:taetracecomp} If $\firstSubtraceSymbol$ and $\secondSubtraceSymbol$ are traces of hybrid programs where $\firstSubtraceSymbol$ terminates and $\text{last }\firstSubtraceSymbol = \text{first } \secondSubtraceSymbol$, then $\tae{\firstSubtraceSymbol \circ \secondSubtraceSymbol}{\formulaSymbol}$ iff both $\tae{\firstSubtraceSymbol}{\formulaSymbol}$ and $\tae{\secondSubtraceSymbol}{\formulaSymbol}$.
\end{lemma}

\begin{proofatend} By \rref{def:dTLTraceSemantics}, since $\firstSubtraceSymbol$ terminates by assumption, $\firstSubtraceSymbol \circ \secondSubtraceSymbol$ is of the form $(\firstSubtraceSymbol_0, \dots, \firstSubtraceSymbol_n, \secondSubtraceSymbol_0, \dots, \secondSubtraceSymbol_m)$. Define $(\firstSubtraceSymbol \circ \secondSubtraceSymbol)_i$ to be $\firstSubtraceSymbol_i$ for $0{\leq}i{\leq}n$ and $\secondSubtraceSymbol_{i - (n+1)}$ for $n+1{\leq}i{\leq}(n + 1) + m$.

First, notice that every $(\firstSubtraceSymbol \circ \secondSubtraceSymbol)_i$ with $|(\firstSubtraceSymbol \circ \secondSubtraceSymbol)_i| = 0$ and $(\firstSubtraceSymbol \circ \secondSubtraceSymbol)_i \neq \Lambda$ satisfies $(\firstSubtraceSymbol \circ \secondSubtraceSymbol)_i(0) \models \cl{\formulaSymbol}$ for $0{\leq}i{\leq}n + m + 1$ iff every $\firstSubtraceSymbol_j$ with $|\firstSubtraceSymbol_j| = 0$ for $0{\leq}j{\leq}n$ satisfies $\firstSubtraceSymbol_j(0) \models \cl{\formulaSymbol}$ and every $\secondSubtraceSymbol_k$ with $|\secondSubtraceSymbol_k| = 0$ and $\secondSubtraceSymbol_k \neq \Lambda$ for $0{\leq}k{\leq}m$ satisfies $\secondSubtraceSymbol_k(0) \models \cl{\formulaSymbol}$. Thus, $\firstSubtraceSymbol \circ \secondSubtraceSymbol$ satisfies the discrete condition of \rref{def:traceformulasemantics} if and only if both $\firstSubtraceSymbol$ and $\secondSubtraceSymbol$ satisfy the discrete condition of \rref{def:traceformulasemantics}.

Next, let $\mathcal{U}$ be the set of positions $(i, \zeta)$ where $(\firstSubtraceSymbol \circ \secondSubtraceSymbol)_i (\zeta) \not \models \formulaSymbol$ and $(\firstSubtraceSymbol \circ \secondSubtraceSymbol)_i (\zeta) \neq \Lambda$, for $0{\leq}i{\leq} n +m + 1$, $0{\leq}\zeta{\leq}|\firstSubtraceSymbol_i|$ when $0{\leq}i{\leq}n$, and $0{\leq}\zeta{\leq}|\secondSubtraceSymbol_i|$ when $n+1{\leq}i{\leq}(n+1) + m$. Then $\mathcal{U} = \mathcal{U}_1 \cup \mathcal{U}_2$ where $\mathcal{U}_1$ is the set of positions $(i, \zeta)$ where $\firstSubtraceSymbol_i(\zeta) \not \models \formulaSymbol$, $0{\leq}i{\leq}n$, $0{\leq}\zeta{\leq}|\firstSubtraceSymbol_i|$ and $\mathcal{U}_2$ is the set of positions $(i + (n + 1), \zeta)$ where $\secondSubtraceSymbol_i(\zeta) \not \models \formulaSymbol$ and $\secondSubtraceSymbol_i(\zeta) \neq \Lambda$, for $0{\leq}i{\leq}m$, $0{\leq}\zeta\leq |\secondSubtraceSymbol_i|$.

Now, consider $f(\mathcal{U}) = f(\mathcal{U}_1 \cup \mathcal{U}_2)$, where $f$ is the mapping as in \rref{def:taef} for $\firstSubtraceSymbol \circ \secondSubtraceSymbol$. Because $\mathcal{U}_1$ and $\mathcal{U}_2$ are disjoint, $f(\mathcal{U}_1 \cup \mathcal{U}_2) = f(\mathcal{U}_1) \cup f(\mathcal{U}_2).$ By construction, $f(\mathcal{U}_1) = f_1(\mathcal{U}_1)$ where $f_1$ is the mapping as in \rref{def:taef} for $\firstSubtraceSymbol$. Now, $f_1(\mathcal{U}_1)$ has measure zero iff $\firstSubtraceSymbol$ satisfies the second condition of \rref{def:traceformulasemantics}.

Also, by construction, $f(\mathcal{U}_2)$ is a translation of $f_2(\widetilde{\mathcal{U}})$ where $f_2$ is the mapping as in \rref{def:taef} for $\secondSubtraceSymbol$ and $\widetilde{\mathcal{U}}$ is the set of positions $(i, \zeta)$ where $\secondSubtraceSymbol_i(\zeta) \not \models \formulaSymbol$ and $\secondSubtraceSymbol_i(\zeta) \neq \Lambda$, for $0{\leq}i{\leq}m$, $0{\leq}\zeta\leq |\secondSubtraceSymbol_i|$.

More precisely, since positions $(j, 0)$ through $(j, |\secondSubtraceSymbol_j|)$ in $\secondSubtraceSymbol$ correspond to positions $((n + 1) + j, 0)$ through $((n + 1) + j, |\secondSubtraceSymbol_j|)$ in $\firstSubtraceSymbol \circ \secondSubtraceSymbol$ for all $0{\leq}j{\leq}m$, the image set $f(\mathcal{U}_2)$ is the image set $f_2(\widetilde{\mathcal{U}})$ translated to the right by $n + 1 + \sum_{i = 0}^n |\firstSubtraceSymbol_i|$ (on the real number line). Now, $f_2(\widetilde{\mathcal{U}})$ has measure zero iff $\secondSubtraceSymbol$ satisfies the continuous condition of \rref{def:traceformulasemantics}. Because the Lebesgue measure is invariant under translation, we also obtain $f(\mathcal{U}_2)$ has measure zero iff $\secondSubtraceSymbol$ satisfies the continuous condition of \rref{def:traceformulasemantics}.

By completeness of the Lebesgue measure, if $f(\mathcal{U})$ has measure zero, then both subsets $f(\mathcal{U}_1)$ and $f(\mathcal{U}_2)$ have measure zero. Since, by additivity of measures, the union of two measure zero sets has measure zero, if both $f(\mathcal{U}_1)$ and $f(\mathcal{U}_2)$ have measure zero, then $f(\mathcal{U})$ has measure zero. So $f(\mathcal{U})$ has measure zero iff both $f(\mathcal{U}_1)$ and $f(\mathcal{U}_2)$ have measure zero. Or, equivalently, $\firstSubtraceSymbol \circ \secondSubtraceSymbol$ satisfies the continuous condition of \rref{def:traceformulasemantics} iff both $\firstSubtraceSymbol$ and $\secondSubtraceSymbol$ satisfy the continuous condition of \rref{def:traceformulasemantics}.

Therefore $\firstSubtraceSymbol \circ \secondSubtraceSymbol$ satisfies both the discrete and continuous conditions of \rref{def:traceformulasemantics} iff both $\firstSubtraceSymbol$ and $\secondSubtraceSymbol$ satisfy the discrete and continuous conditions of \rref{def:traceformulasemantics}, and thus we have $\tae{\firstSubtraceSymbol \circ \secondSubtraceSymbol}{\formulaSymbol}$ iff both $\tae{\firstSubtraceSymbol}{\formulaSymbol}$ and $\tae{\secondSubtraceSymbol}{\formulaSymbol}$.
 \end{proofatend}

\subsection{Proof Calculus}\label{sec:proofcalculus} The proof calculus for \PdTL is shown in \rref{fig:proofRulesTable}.
Intuitively, all of the axioms are designed to successively decompose complicated formulas into structurally simpler formulas while successively reducing trace formulas into state formulas.
The test axiom, assignment axiom, solution axiom, and solution with evolution domain constraint axiom ($[\text{?}]_{\text{tae}}$, $[:=]_{\text{tae}}$, $[']_{\text{tae}}$, and $['\&]_{\text{tae}}$) remove instances of $\Box_{\text{tae}}$.
The nondeterministic choice axiom $[\cup]_{\text{tae}}$ reduces a choice between two hybrid programs to two separate programs.  The induction axiom $\text{I}_{\text{tae}}$ reduces a loop property involving a trace formula to a loop property involving a state formula; $\text{I}_{\text{tae}}$ also allows us to derive two very useful proof rules.

The G\"{o}del generalization rule ($\text{G}_{\text{tae}}$) proves that if formula $\formulaSymbol$ is valid, then it is also true, tae, along the trace of any hybrid program. The modal modus ponens rule ($K_{\text{tae}}$) allows us to derive a monotonicity property.
Our approach occasionally introduces extra premises; for example, the modal modus ponens rule ($K_{\text{tae}}$) has an extra goal $\cl{\formulaSymbol} \to \cl{\secondFormulaSymbol}$ due to the discrete condition of \rref{def:traceformulasemantics}. 
Many of these extra premises will be easy to prove---if our models make use of physical formulas, which are closed, then these extra cases will prove immediately.

\begin{figure}[tbhp]
  \def\arraystretch{1.2}
    \begin{tabular}{l l@{\hspace{0.5cm}} c}
    $[\text{?}]_{\text{tae}}$ & $\rae{?\FOLformulaSymbol}{\formulaSymbol} \leftrightarrow \cl{\formulaSymbol}$ & {\begin{prooftree}
    \hypo{\formulaSymbol}
    \infer1[$\text{G}_{\text{tae}}$]{\rae{\alpha}{\formulaSymbol}}
    \end{prooftree}}
    \\
        $[\cup]_{\text{tae}}$ & $\rae{\alpha \cup \beta}{\formulaSymbol} \leftrightarrow \rae{\alpha}{\formulaSymbol} \land \rae{\beta}{\formulaSymbol}$ & \\ 
        $[:=]_{\text{tae}}$ & $\rae{x := e}{\formulaSymbol} \leftrightarrow \cl{\formulaSymbol} \land [x := e]\cl{\formulaSymbol}$ & {  \begin{prooftree}
    \hypo{\cl{\formulaSymbol} \to \cl{\secondFormulaSymbol}}
    \hypo{\rae{\alpha}{(\formulaSymbol \to \secondFormulaSymbol)}}
    \infer2[$K_{\text{tae}}$]{\rae{\alpha}{\formulaSymbol} \to \rae {\alpha}{\secondFormulaSymbol}}
    \end{prooftree}}
        \\ 
        $[;]_{\text{tae}}$ & $\rae{\alpha; \beta}{\formulaSymbol} \leftrightarrow \left( \rae{\alpha}{\formulaSymbol} \land [\alpha] \rae{\beta}{\formulaSymbol}\right)$ &  \\
        $\text{I}_{\text{tae}}$ & $ \rae{\alpha^*}{\formulaSymbol} \leftrightarrow \left(\cl{\formulaSymbol} \land [\alpha^*](\mrae{\cl{\formulaSymbol}}{\alpha}{\formulaSymbol})\right)$ & {\begin{prooftree} \hypo{\formulaSymbol \to \secondFormulaSymbol} \infer1[TopCl]{\cl{\formulaSymbol} \to \cl{\secondFormulaSymbol}} \end{prooftree}}\\
       $[']_{\text{tae}}$ & $\rae{x' = f(x)}{\FOLformulaSymbol} \leftrightarrow \cl{\FOLformulaSymbol} \land \forall t{\geq}0 Q$ &
      \vspace{.5em} \\  
      $['\&]_{\text{tae}}$ & $\rae{x' = f(x) \& \evDomainConstraint}{\FOLformulaSymbol} \leftrightarrow \cl{\FOLformulaSymbol} \land$ & CGG\ \ $\rae{\alpha}{\formulaSymbol} \to [\alpha]\cl{\formulaSymbol}$ \\
 & \hspace{2em} $\forall t {>}0 \left(\left( \forall 0{\leq}s{\leq}t \ [x := y(s)]\evDomainConstraint \right) \to \secondFOLformulaSymbol \right)$ &
    \end{tabular}
    \caption[Proof calculus for \PdTL]{Proof calculus for \PdTL\footnotemark}
    \label{fig:proofRulesTable}
\end{figure}
   \footnotetext{Here, $\alpha$ and $\beta$ are hybrid programs, $\formulaSymbol$ and $\secondFormulaSymbol$ are state formulas, $\FOLformulaSymbol$ is a FOL formula, $y(t)$ is the unique global polynomial solution to the differential equation $x' = f(x)$, and the formula $Q$ in $[']_{\text{tae}}$ and $['\&]_{\text{tae}}$ is the FOL formula constructed by \rref{prop:definable} for $P(y(t))$. Although the ``for almost all'' quantifier is in general \textit{not} definable in FOL \cite{DBLP:journals/jsyml/Morgenstern79}, \rref{prop:definable} justifies that ``$\text{for almost all }t{\geq}0 [x := y(t)] \FOLformulaSymbol$'' is logically equivalent to ``$\forall t{\geq}0\ Q$''.}

TopCl is from \rref{prop:clProp1} and CGG from \rref{cor:EM}.
The rest of the soundness proofs are in
\rref{app:ProofCalculusSoundness}.
We discuss a few key high-level ideas below.

\subsubsection{Soundness of $[;]_{\text{tae}}$ and $\text{I}_{\text{tae}}$} Sequential composition and induction are subtly challenging for PHS---since we are allowed to leave the safe set, handover points between hybrid programs are no longer guaranteed to be safe points.
However, sequential composition and induction are crucial for the practicality of verification, which is predicated on having a good way of breaking down complicated formulas into simpler components.
The soundness proof of $[;]_{\text{tae}}$ exploits \rref{lem:taetracecomp}.
The soundness proof of $\text{I}_{\text{tae}}$ is based on \rref{lem:laststatelemma}, which in turn relies on the discrete condition of \rref{def:traceformulasemantics}.

\subsubsection{Differential Equations} Reasoning about differential equations is one of the most challenging aspects of hybrid systems. In this work, we focus on relatively simple reasoning principles for differential equations, as justifying even simple principles is made much more challenging by introducing the notion of ``safety almost everywhere''. We leave the development of more complicated reasoning (for example, a notion of differential invariants for $\Box_{\text{tae}}$) to future work.

For ``sufficiently tame'' systems of ODEs $x' = f(x)$, we might hope to replace $\rae{x' = f(x)}{\FOLformulaSymbol}$ with an equivalent expression without the $\Box_{\text{tae}}$ modality. The cleanest case is when $x' = f(x)$ has a unique global polynomial solution, $y(t)$. 
Although we think of $y$ as being a polynomial in $t$, $y$ can involve any of the other parameters, call them $x_1, \dots, x_n$, in $f$, from its dependency on initial values. 
We require that $y$ is also polynomial in $x_1, \dots, x_n$.
This is the case handled by axiom $[']_{\text{tae}}$: $\rae{x' = f(x)}{\FOLformulaSymbol} \leftrightarrow \cl{\FOLformulaSymbol}\ \land\ \forall t{\geq}0\, Q$, where $\secondFOLformulaSymbol$ is a FOL formula constructed so that ``$\forall t{\geq} 0 \, \secondFOLformulaSymbol$'' expresses ``$\text{for almost all } t{\geq}0\, [x := y(t)] \FOLformulaSymbol$''. Axiom $['\&]_{\text{tae}}$ generalizes $[']_{\text{tae}}$ to ODEs with evolution domain constraints.

In particular, we get $\secondFOLformulaSymbol$ by applying \rref{prop:definable} to $\FOLformulaSymbol(y(t))$, which is the formula obtained using the assignment axiom $[:=]$ of \dL on $[x := y(t)] \FOLformulaSymbol$ (it is a FOL formula because $y(t)$ is polynomial and polynomials are closed under composition).
Given any FOL formula $\FOLformulaSymbol$, \rref{prop:definable} constructs a FOL formula $\secondFOLformulaSymbol$ so that ``for almost all $t{\geq}0\, \FOLformulaSymbol$'' is semantically equivalent to $\forall t{\geq} 0\, \secondFOLformulaSymbol$ (i.e., ``for almost all $t{\geq}0\, \FOLformulaSymbol$'' is true in a state $\firstStateSymbol$ iff $\forall t{\geq}0\, \secondFOLformulaSymbol$ is true in $\firstStateSymbol$).
\begin{proposition}\label{prop:definable}
Let $\FOLformulaSymbol$ be a FOL formula.  Using quantifier elimination \cite{Ta}, put it into one of the following normal forms: $e{=}0$, $e{\geq}0$, $e{<}0$, $\FOLformulaSymbol_1 \land \FOLformulaSymbol_2$, and $\FOLformulaSymbol_1 \lor \FOLformulaSymbol_2$, where $e$ is a polynomial and $\FOLformulaSymbol_1, \FOLformulaSymbol_2$ are FOL formulas.  Construct the FOL formula $\secondFOLformulaSymbol = g(P)$ by structural induction on $P$ as follows: $g(e{=}0)$ is $e{=}0$,  
$g(e{\geq}0)$ is $e{\geq}0$, $g(e{<}0)$ is $e{\leq}0 \land ((a_n{=}0 \land \cdots \land a_1{=}0) \to e{<}0)$, $g(\FOLformulaSymbol_1 \land \FOLformulaSymbol_2)$ is $g(\FOLformulaSymbol_1) \land g(\FOLformulaSymbol_2)$, and $g(\FOLformulaSymbol_1 \lor \FOLformulaSymbol_2)$ is $g(\FOLformulaSymbol_1) \lor g(\FOLformulaSymbol_2)$.

Then, for any state $\firstStateSymbol$, the following hold:
\begin{enumerate}
    \item Locally false: If $\updateState{\firstStateSymbol}{t}{k} \not \models \secondFOLformulaSymbol$ for some $k{\geq}0$, then there is a nonempty interval $[k, \ell)$ so that for all $q \in [k, \ell)$, $\updateState{\firstStateSymbol}{t}{q} \not \models \FOLformulaSymbol$. 
    Further, if $k {>} 0$, then there is an interval $(\ell_1, \ell_2)$ with $\ell_1{<}k{<}\ell_2$ so that for all $q \in (\ell_1, \ell_2)$, $\updateState{\firstStateSymbol}{t}{q} \not \models \FOLformulaSymbol$.
    \item Finite difference: There are only finitely many values $k{\geq}0$ where $\updateState{\firstStateSymbol}{t}{k} \models \secondFOLformulaSymbol \land \lnot \FOLformulaSymbol$.
\end{enumerate}
\end{proposition}

The proof is by induction on the structure of $\FOLformulaSymbol$.
Details are given in \rref{app:ProofCalculusSoundness}. 
The soundness proofs of $[']_{\text{tae}}$ and $['\&]_{\text{tae}}$ are also in \rref{app:ProofCalculusSoundness}.

\subsection{Derived Rules}
We highlight some of the most useful derived rules for \PdTL formulas in \rref{fig:derivedRulesTable}.
Monotonicity properties are fundamental in logic. Our rule $M_{\text{tae}}$ intuitively says that if $\secondFormulaSymbol \to \formulaSymbol$ is valid, then if $\secondFormulaSymbol$ is true almost everywhere along every trace of a hybrid program, then $\formulaSymbol$ is also true almost everywhere along every trace of that hybrid program.
The rule $\text{Ind}_{\text{tae}}$ reduces proving a safety property of hybrid program $\alpha^*$ to proving a safety property of program $\alpha$.
When its premise proves, it effectively removes the need to reason about loops.
The rule $\text{loop}_{\text{tae}}$ provides us with a loop invariant rule.
The rule $\text{Comp}_{\text{tae}}$ reduces a property of $\alpha; \beta$ to individual properties of $\alpha$ and $\beta$.
The derivations are given in \rref{app:DerivedSoundness}.

\begin{figure}[tbhp]
    \begin{tabular}{l@{\hspace{0.3cm}} c}
    {\begin{prooftree}
    \hypo{\secondFormulaSymbol \to \formulaSymbol}
    \infer1[$\text{M}_{\text{tae}}$]{\rae{\alpha}{\secondFormulaSymbol} \to \rae {\alpha}{\formulaSymbol}}
    \end{prooftree}}  &  {\begin{prooftree}
    \hypo{\cl{\formulaSymbol} \vdash \rae{\alpha}{\formulaSymbol}}
    \infer1[$\text{Ind}_{\text{tae}}$]{\cl{\formulaSymbol} \vdash \rae {\alpha^*}{\formulaSymbol}}
    \end{prooftree}}\vspace{.6em} \\
    {\begin{prooftree}
\hypo{\Gamma \vdash \cl{\secondFormulaSymbol}, \Delta}
\hypo{\cl{\secondFormulaSymbol} \vdash \rae{\alpha}{\secondFormulaSymbol}}
\hypo{\secondFormulaSymbol \vdash \formulaSymbol}
\infer3[$\text{loop}_{\text{tae}}$]{\Gamma \vdash \rae{\alpha^*}{\formulaSymbol}, \Delta}
\end{prooftree}} & {\begin{prooftree}
\hypo{\secondFormulaSymbol \to \rae {\alpha}{\formulaSymbol}}
\hypo{\cl{\formulaSymbol} \to \rae {\beta}{\formulaSymbol}}
\infer2[$\text{Comp}_{\text{tae}}$]{\mrae{\secondFormulaSymbol}{\alpha; \beta}{\formulaSymbol}}
\end{prooftree}}
        \end{tabular}
    \caption{Derived rules for \PdTL}
    \label{fig:derivedRulesTable}
\end{figure}

\section{Proof of Motivating Example}\label{sec:Example} 
We now apply our proof calculus to the model of the train example (\rref{sec:syntax}).
Full details are in \rref{app:Example}.
Using structural rule ${\to}R$ and our induction proof rule $\text{loop}_{\text{tae}}$ with invariant $v{<}100$, the proof reduces to showing $a{=}0 \land v{=}0 \vdash v{\leq}100$ (which holds by real arithmetic), $v{<}100 \vdash v{<}100$ (identically true), and 
\begin{align*}v{\leq}100 \vdash &[\big((?(v{<}100); a := 1) \cup (?(v = 100); a := -1)\big); \\
 & \hspace{2em} \{x' = v, v' = a \ \&\  0{\leq}v{\leq}100\}] \Box_{\text{tae}} v{<}100.
\end{align*}
Axiom $[;]_{\text{tae}}$ splits this into goals (\rref{eqn:MG1p}) and (\rref{eqn:MG2p}):
\begin{equation}\label{eqn:MG1p}
 v{\leq}100 \vdash [(?(v{<}100); a := 1) \cup (?(v = 100); a := -1)]\Box_{\text{tae}} v{<}100
 \end{equation}
 \\
\begin{equation}\label{eqn:MG2p}
\begin{split}
v{\leq}100 \vdash [(?(v{<}100); a := 1) &\cup (?(v = 100); a := -1)] \\
[ \{x' = v, v' = a 
& \&\  0{\leq}v{\leq}100\}]\Box_{\text{tae}} v{<}100.
\end{split}
\end{equation}
(\rref{eqn:MG1p}) is straighforward.
(\rref{eqn:MG2p}) is more complicated because it involves ODEs reasoning.
The \dL axioms $[\cup]$ and $\land R$ split the proof of (\rref{eqn:MG2p}) into (\rref{eqn:SG3p}) and (\rref{eqn:SG4p}):
\begin{align} \label{eqn:SG3p}
v{\leq}100 &\vdash [?(v{<}100); a := 1] [ \{x' = v, v' = a\ \&\  0{\leq}v{\leq}100\}]\Box_{\text{tae}} v{<}100
\\
\label{eqn:SG4p}
    v{\leq}100 &\vdash [?(v = 100); a := -1] [ \{x' = v, v' = a\ \&\  0{\leq}v{\leq}100\}]\Box_{\text{tae}} v{<}100
\end{align}
(\rref{eqn:SG3p}) and (\rref{eqn:SG4p}) require similar reasoning, so we focus on (\rref{eqn:SG3p}).
The \dL axioms $[;]$, $[:=]$, and $[?]$ reduce (\rref{eqn:SG3p}) to 
\begin{equation} \label{eqn:ODE1p}
v{\leq}100, v{<}100 \vdash [ \{x' = v, v' = 1\ \&\  0{\leq}v{\leq}100\}]\Box_{\text{tae}} v{<}100
\end{equation}
To prove (\rref{eqn:ODE1p}), we need to use axiom $['\&]_{\text{tae}}$, which says:
\begin{equation*}
    \rae{x' = f(x) \& \evDomainConstraint}{\FOLformulaSymbol} \leftrightarrow \cl{\FOLformulaSymbol} \land \forall t {>}0 \left(\left( \forall 0{\leq}s{\leq}t \ [x := y(s)]\evDomainConstraint \right) \to \secondFOLformulaSymbol \right)
\end{equation*}
For clarity, we use $v_0$ for the value of $v$ in the initial state before it starts evolving along the ODEs, and similarly we use $x_0$ for the value of $x$ in the initial state.
Following \rref{prop:definable}, $Q$ is $(1 = 0 \to t + v_0{<}100) \land t + v_0 {\leq}100$. Since applying $['\&]_{tae}$ reduces our goal to a \dL formula, and since \PdTL is a conservative extension of \dL, we can use the contextual equivalence rules of \dL to replace $Q$ with the logically equivalent formula $t+v_0{\leq}100$, obtaining:
\begin{equation}\label{eqn:DL1p}
\begin{split}
    & v_0{\leq}100, v_0{<}100 \vdash  v_0{\leq}100 \land \forall t {>}0 \\
    &\left(\left( \forall 0{\leq}s{\leq}t \ [x := .5s^2 + v_0s + x_0][v := s + v_0] 0{\leq}v{\leq}100 \right) \to t+v_0{\leq}100 \right)
    \end{split}
\end{equation}
After using the \dL axiom $[:=]$, the proof closes by real arithmetic.

\section{Conclusions and Future Work}\label{sec:FutureWork}
We introduce PHS to help narrow the gap between mathematical models and physical reality.
To enable logic to begin to distinguish between true unsafeties of systems and physically unrealistic unsafeties, we develop the notion of safety tae along the execution trace of a system.
Our new logic, \PdTL, contains the logical operator $\Box_{\text{tae}}$, which elides sets of time that have measure zero.

A cornerstone of our approach is its logical practicality---in order to support verification, we develop a proof calculus for \PdTL.
We demonstrate the capability of the proof calculus by applying it to a motivating example.
We think it is an interesting and challenging problem for future work to develop new ways of thinking about PHS, such as the notion of space almost everywhere discussed in \rref{sec:discussion}, while maintaining this logical practicality.

Future work could continue to develop \PdTL.
It would be especially interesting to develop further differential equations reasoning, including an appropriate generalization of the syntax of hybrid programs to admit Carath\'{e}odory solutions. 

\section*{Acknowledgments}
We very much appreciate Yong Kiam Tan and Brandon Bohrer for many useful discussions and for feedback on the paper. Thank you also to the anonymous CADE'19 reviewers for their thorough feedback.

This material is based upon work supported by the National Science Foundation Graduate Research Fellowship under Grant No. DGE-1252522. Any opinions, findings, and conclusions or recommendations expressed in this material are those of the authors and do not necessarily reflect the views of the National Science Foundation. This research was also sponsored by the AFOSR under grant number FA9550-16-1-0288 and by the Alexander von Humboldt Foundation. The views and conclusions contained in this document are those of the authors and should not be interpreted as representing the official policies, either expressed or implied, of any sponsoring institution, the U.S. government or any other entity.

\renewcommand{\doi}[1]{doi:\href{https://doi.org/#1}{\nolinkurl{#1}}}
\bibliographystyle{plainurl}
\bibliography{PHS}

\begin{thebibliography}{10}

\bibitem{DBLP:conf/hybrid/AlurCHH92}
Rajeev Alur, Costas Courcoubetis, Thomas~A. Henzinger, and Pei-Hsin Ho.
\newblock Hybrid automata: An algorithmic approach to the specification and
  verification of hybrid systems.
\newblock In Grossman et~al. \cite{DBLP:conf/hybrid/1992}, pages 209--229.

\bibitem{DBLP:conf/lics/AsarinB01}
Eugene Asarin and Ahmed Bouajjani.
\newblock Perturbed turing machines and hybrid systems.
\newblock In {\em 16th Annual {IEEE} Symposium on Logic in Computer Science,
  Boston, Massachusetts, USA, June 16-19, 2001, Proceedings}, pages 269--278.
  {IEEE} Computer Society, 2001.
\newblock \href {https://doi.org/10.1109/LICS.2001.932503}
  {\path{doi:10.1109/LICS.2001.932503}}.

\bibitem{BCR}
Jacek Bochnak, Michel Coste, and Marie-Fran{\c{c}}oise Roy.
\newblock {\em Real Algebraic Geometry}.
\newblock Springer, 1998.
\newblock \href {https://doi.org/10.1007/978-3-662-03718-8}
  {\path{doi:10.1007/978-3-662-03718-8}}.

\bibitem{DBLP:journals/scl/Ceragioli02}
Francesca Ceragioli.
\newblock Some remarks on stabilization by means of discontinuous feedbacks.
\newblock {\em Systems {\&} Control Letters}, 45(4):271--281, 2002.
\newblock \href {https://doi.org/10.1016/S0167-6911(01)00185-2}
  {\path{doi:10.1016/S0167-6911(01)00185-2}}.

\bibitem{Cortes2008}
Jorge Cortes.
\newblock Discontinuous dynamical systems.
\newblock {\em IEEE Control Systems}, 28(3):36--73, 2008.
\newblock \href {https://doi.org/10.1109/MCS.2008.919306}
  {\path{doi:10.1109/MCS.2008.919306}}.

\bibitem{dam}
Mads Dam.
\newblock {CTL}* and {ECTL}* as fragments of the modal $\mu$-calculus.
\newblock {\em Theoretical Computer Science}, 126(1):77--96, 1994.
\newblock \href {https://doi.org/10.1016/0304-3975(94)90269-0}
  {\path{doi:10.1016/0304-3975(94)90269-0}}.

\bibitem{DBLP:conf/cav/DonzeFM13}
Alexandre Donz{\'{e}}, Thomas Ferr{\`{e}}re, and Oded Maler.
\newblock Efficient robust monitoring for {STL}.
\newblock In Natasha Sharygina and Helmut Veith, editors, {\em Computer Aided
  Verification - 25th International Conference, {CAV} 2013, Saint Petersburg,
  Russia, July 13-19, 2013. Proceedings}, volume 8044 of {\em LNCS}, pages
  264--279. Springer, 2013.
\newblock \href {https://doi.org/10.1007/978-3-642-39799-8\_19}
  {\path{doi:10.1007/978-3-642-39799-8\_19}}.

\bibitem{STL}
Georgios~E. Fainekos and George~J. Pappas.
\newblock Robustness of temporal logic specifications for continuous-time
  signals.
\newblock {\em Theoretical Computer Science}, 410(42):4262--4291, 2009.
\newblock \href {https://doi.org/10.1016/j.tcs.2009.06.021}
  {\path{doi:10.1016/j.tcs.2009.06.021}}.

\bibitem{DBLP:conf/csl/Franzle99}
Martin Fr{\"{a}}nzle.
\newblock Analysis of hybrid systems: An ounce of realism can save an infinity
  of states.
\newblock In J{\"{o}}rg Flum and Mario Rodr{\'{\i}}guez{-}Artalejo, editors,
  {\em Computer Science Logic, 13th International Workshop, {CSL} '99, 8th
  Annual Conference of the EACSL, Madrid, Spain, September 20-25, 1999,
  Proceedings}, volume 1683 of {\em LNCS}, pages 126--140. Springer, 1999.
\newblock \href {https://doi.org/10.1007/3-540-48168-0\_10}
  {\path{doi:10.1007/3-540-48168-0\_10}}.

\bibitem{DBLP:conf/lics/GaoAC12}
Sicun Gao, Jeremy Avigad, and Edmund~M. Clarke.
\newblock Delta-decidability over the reals.
\newblock In {\em Proceedings of the 27th Annual {IEEE} Symposium on Logic in
  Computer Science, {LICS} 2012, Dubrovnik, Croatia, June 25-28, 2012}, pages
  305--314. {IEEE} Computer Society, 2012.
\newblock \href {https://doi.org/10.1109/LICS.2012.41}
  {\path{doi:10.1109/LICS.2012.41}}.

\bibitem{HDS}
R.~{Goebel}, R.~G. {Sanfelice}, and A.~R. {Teel}.
\newblock Hybrid dynamical systems.
\newblock {\em IEEE Control Systems Magazine}, 29(2):28--93, April 2009.
\newblock \href {https://doi.org/10.1109/MCS.2008.931718}
  {\path{doi:10.1109/MCS.2008.931718}}.

\bibitem{GOEBEL20041}
Rafal Goebel, Joao Hespanha, Andrew~R. Teel, Chaohong Cai, and Ricardo
  Sanfelice.
\newblock Hybrid systems: Generalized solutions and robust stability.
\newblock {\em IFAC Proceedings Volumes}, 37(13):1 -- 12, 2004.
\newblock 6th IFAC Symposium on Nonlinear Control Systems 2004 (NOLCOS 2004),
  Stuttgart, Germany, 1-3 September, 2004.
\newblock \href {https://doi.org/10.1016/S1474-6670(17)31194-1}
  {\path{doi:10.1016/S1474-6670(17)31194-1}}.

\bibitem{DBLP:conf/hybrid/1992}
Robert~L. Grossman, Anil Nerode, Anders~P. Ravn, and Hans Rischel, editors.
\newblock {\em Hybrid Systems}, volume 736, Berlin, 1993. Springer.

\bibitem{DBLP:conf/cade/JeanninP14}
Jean{-}Baptiste Jeannin and Andr{\'{e}} Platzer.
\newblock {dTL$^2$}: Differential temporal dynamic logic with nested
  temporalities for hybrid systems.
\newblock In St{\'{e}}phane Demri, Deepak Kapur, and Christoph Weidenbach,
  editors, {\em Automated Reasoning - 7th International Joint Conference,
  {IJCAR} 2014, Held as Part of the Vienna Summer of Logic, {VSL} 2014, Vienna,
  Austria, July 19-22, 2014. Proceedings}, volume 8562 of {\em LNCS}, pages
  292--306. Springer, 2014.
\newblock \href {https://doi.org/10.1007/978-3-319-08587-6\_22}
  {\path{doi:10.1007/978-3-319-08587-6\_22}}.

\bibitem{M0Jeffreys}
Harold Jeffreys and Bertha Swirles.
\newblock {\em Methods of {M}athematical {P}hysics}.
\newblock Cambridge University Press, Cambridge, 3rd edition, 1999.
\newblock \href {https://doi.org/10.1017/CBO9781139168489}
  {\path{doi:10.1017/CBO9781139168489}}.

\bibitem{DBLP:conf/tacas/KongGCC15}
Soonho Kong, Sicun Gao, Wei Chen, and Edmund~M. Clarke.
\newblock d{R}each: {\(\delta\)}-reachability analysis for hybrid systems.
\newblock In Christel Baier and Cesare Tinelli, editors, {\em Tools and
  Algorithms for the Construction and Analysis of Systems - 21st International
  Conference, {TACAS} 2015, Held as Part of the European Joint Conferences on
  Theory and Practice of Software, {ETAPS} 2015, London, UK, April 11-18, 2015.
  Proceedings}, volume 9035 of {\em LNCS}, pages 200--205. Springer, 2015.
\newblock \href {https://doi.org/10.1007/978-3-662-46681-0\_15}
  {\path{doi:10.1007/978-3-662-46681-0\_15}}.

\bibitem{MANTHANWAR20051249}
A.M. Manthanwar, V.~Sakizlis, V.~Dua, and E.N. Pistikopoulos.
\newblock Robust model-based predictive controller for hybrid system via
  parametric programming.
\newblock In Luis Puigjaner and Antonio Espu{\~n}a, editors, {\em European
  Symposium on Computer-Aided Process Engineering-15, 38th European Symposium
  of the Working Party on Computer Aided Process Engineering}, volume~20 of
  {\em Computer Aided Chemical Engineering}, pages 1249 -- 1254. Elsevier,
  2005.
\newblock \href {https://doi.org/10.1016/S1570-7946(05)80050-1}
  {\path{doi:10.1016/S1570-7946(05)80050-1}}.

\bibitem{mayhew1}
C.~G. {Mayhew}, R.~G. {Sanfelice}, and A.~R. {Teel}.
\newblock Robust source-seeking hybrid controllers for autonomous vehicles.
\newblock In {\em 2007 American Control Conference}, pages 1185--1190, July
  2007.
\newblock \href {https://doi.org/10.1109/ACC.2007.4283016}
  {\path{doi:10.1109/ACC.2007.4283016}}.

\bibitem{DBLP:journals/cj/McCallum93}
Scott McCallum.
\newblock Solving polynomial strict inequalities using cylindrical algebraic
  decomposition.
\newblock {\em The Computer Journal}, 36(5):432--438, 1993.
\newblock \href {https://doi.org/10.1093/comjnl/36.5.432}
  {\path{doi:10.1093/comjnl/36.5.432}}.

\bibitem{DBLP:journals/tcs/MoggiFDT18}
Eugenio Moggi, Amin Farjudian, Adam Duracz, and Walid Taha.
\newblock Safe {\&} robust reachability analysis of hybrid systems.
\newblock {\em Theoretical Computer Science}, 747:75--99, 2018.
\newblock \href {https://doi.org/10.1016/j.tcs.2018.06.020}
  {\path{doi:10.1016/j.tcs.2018.06.020}}.

\bibitem{DBLP:conf/hybrid/MoorD01}
Thomas Moor and Jennifer~M. Davoren.
\newblock Robust controller synthesis for hybrid systems using modal logic.
\newblock In Maria Domenica~Di Benedetto and Alberto~L.
  Sangiovanni{-}Vincentelli, editors, {\em Hybrid Systems: Computation and
  Control, 4th International Workshop, {HSCC} 2001, Rome, Italy, March 28-30,
  2001, Proceedings}, volume 2034 of {\em LNCS}, pages 433--446. Springer,
  2001.
\newblock \href {https://doi.org/10.1007/3-540-45351-2\_35}
  {\path{doi:10.1007/3-540-45351-2\_35}}.

\bibitem{DBLP:journals/jsyml/Morgenstern79}
Carl~F. Morgenstern.
\newblock The measure quantifier.
\newblock {\em J. Symb. Log.}, 44(1):103--108, 1979.
\newblock \href {https://doi.org/10.2307/2273708} {\path{doi:10.2307/2273708}}.

\bibitem{DBLP:conf/hybrid/NerodeK92a}
Anil Nerode and Wolf Kohn.
\newblock Models for hybrid systems: Automata, topologies, controllability,
  observability.
\newblock In Grossman et~al. \cite{DBLP:conf/hybrid/1992}, pages 317--356.

\bibitem{Pl}
Andr{\'{e}} Platzer.
\newblock {\em Logical Analysis of Hybrid Systems - Proving Theorems for
  Complex Dynamics}.
\newblock Springer, Heidelberg, 2010.
\newblock \href {https://doi.org/10.1007/978-3-642-14509-4}
  {\path{doi:10.1007/978-3-642-14509-4}}.

\bibitem{Pl2}
Andr{\'{e}} Platzer.
\newblock {\em Logical Foundations of Cyber-Physical Systems}.
\newblock Springer, Cham, 2018.
\newblock \href {https://doi.org/10.1007/978-3-319-63588-0}
  {\path{doi:10.1007/978-3-319-63588-0}}.

\bibitem{DBLP:conf/lics/PlatzerT18}
Andr{\'{e}} Platzer and Yong~Kiam Tan.
\newblock Differential equation axiomatization: The impressive power of
  differential ghosts.
\newblock In Anuj Dawar and Erich Gr{\"{a}}del, editors, {\em LICS}, pages
  819--828, New York, 2018. ACM.
\newblock \href {https://doi.org/10.1145/3209108.3209147}
  {\path{doi:10.1145/3209108.3209147}}.

\bibitem{Royden}
Halsey~Lawrence Royden and Patrick~M. Fitzpatrick.
\newblock {\em Real Analysis (Classic Version)}.
\newblock Pearson, 4th edition, 2018.

\bibitem{Ta}
Alfred Tarski.
\newblock A decision method for elementary algebra and geometry.
\newblock In Bob~F. Caviness and Jeremy~R. Johnson, editors, {\em Quantifier
  Elimination and Cylindrical Algebraic Decomposition}, pages 24--84, Vienna,
  1998. Springer.
\newblock \href {https://doi.org/10.1007/978-3-7091-9459-1\_3}
  {\path{doi:10.1007/978-3-7091-9459-1\_3}}.

\bibitem{walter}
Wolfgang Walter.
\newblock {\em Ordinary Differential Equations}, volume 182 of {\em Graduate
  Texts in Mathematics}.
\newblock Springer, New York, 1998.
\newblock \href {https://doi.org/10.1007/978-1-4612-0601-9}
  {\path{doi:10.1007/978-1-4612-0601-9}}.

\end{thebibliography}

\appendix
\appendix
\section{Measure Theory Background}\label{app:MeasureTheory}
Although our paper does not require a detailed understanding of measure theory, we do make use of certain properties of the Lebesgue measure and of measure zero.
For the interested reader, there are many good books on measure theory; for example \cite{Royden}.
We give a very short overview of the basics of measure theory that we require.

Intuitively, a \textit{measure} on a set $X$ maps subsets of $X$ to real numbers (or to infinity) in a way designed to capture information regarding the ``size'' of these subsets. Perhaps the most familiar measures are length, area, and volume (on $\mathbb{R}$, $\mathbb{R}^2$, and $\mathbb{R}^3$, respectively). Measures are always nonnegative, so that if $\mu$ is a measure on set $X$ and $A$ is a measurable subset of $X$, then $\mu(A) \geq 0$. Further, measures are countably additive, so that the measure of a countable union of subsets is the sum of the measures of each subset---i.e. if $\mu$ is a measure on $X$ and $A_i \subseteq X$ is measurable for every $i$, then $\mu(\cup_{i = 1}^{\infty} A_i) = \sum_{i = 1}^{\infty} \mu(A_i)$. Finally, the measure of the empty set is zero. In order to satisfy these properties, it may be necessary that some subsets of $X$ do not have an associated measure (Banach-Tarski Paradox).
In particular, not every subset of a set that has a measure is guaranteed to have a measure itself.

Our paper uses the \textit{Lebesgue measure} $\mu$ on $\mathbb{R}$.
This measure is designed to rigorize the natural notion of length.
Every $\mathbb{R}^n$ has an associated Lebesgue measure.
The Lebesgue measures on $\mathbb{R}^2$ and $\mathbb{R}^3$ rigorize the notions of area and volume, respectively.
We are only interested in the case where $n = 1$, since we are mapping traces of hybrid programs to a single time axis (see Section \ref{sec:Semantics}).
Most reasonable subsets of $\mathbb{R}$ have an associated Lebesgue measure. Notably, any interval $(a, b)$ with $a < b$ has nonzero Lebesgue measure $a - b \neq 0$ in $\mathbb{R}$ and likewise for $[a,b]$.

We focus in particular on Lebesgue measure zero. A subset $\mathcal{U}$ of $\mathbb{R}$ has \textit{Lebesgue measure zero} iff for all $\epsilon {>} 0$ there exist intervals $I_p = [a_p, b_p]$ so that $\mathcal{U} \subseteq \bigcup_{p = 0}^{\infty} I_p$ and $\sum_{p = 0}^{\infty} |b_p - a_p|{<}\epsilon$ (see \cite[Section 1-1013]{M0Jeffreys}).
Measure zero sets satisfy many nice properties.
Since measures are additive, the union of two measure zero sets has measure zero; i.e. if $\mu(A) = 0$ and $\mu(B) = 0$, then $\mu(A \cup B) = \mu(A) + \mu(B) = 0 + 0 = 0$.
Also, the Lebesgue measure on $\mathbb{R}$ is \textit{complete}, i.e., each subset of any set of measure zero is itself measurable and has measure zero.
To see this, consider $\mathcal{S}_1 \subset \mathbb{R}$ with Lebesgue measure zero and $\mathcal{S}_2 \subseteq \mathcal{S}_1$.
Then, to see that $\mathcal{S}_2$ is measurable and has Lebesgue measure zero, simply choose the same intervals $I_p$ for $\mathcal{S}_2$ as for $\mathcal{S}_1$ in the definition of measure zero.
Finally, notice that if $(a, b) \subseteq \mathcal{S}$ for any subset $\mathcal{S}\subseteq\mathbb{R}$ of the reals and any interval $(a, b)$, then $\mathcal{S}$ cannot have measure zero.

\section{Conservativity}\label{app:Conservativity}
We first show that we have a \textit{conservative} extension of \dL: i.e. that all valid \dL formulas are valid in \PdTL.
In \dL, the semantics of differential equations is as follows \cite{Pl,Pl2}:

\begin{definition}\label{def:classicalDef} State $\secondStateSymbol$ is \textit{reachable} from initial state $\firstStateSymbol$ by $x_1'{=}\theta_1, \dots, x_n'{=}\theta_n\ \&\ \evDomainConstraint$ iff there is a function $\varphi : [0,r] \to \stateSet{\realValuedVariableSet}$ such that:
\begin{enumerate}
\item Initial and final states match: $\varphi(0) = \firstStateSymbol, \varphi(r) = \secondStateSymbol$.
\item $\varphi$ respects the differential equations: For each variable $x_i$, $\evalInState{\varphi(z)}{x_i}$ is continuous in $z$ on $[0, r]$ and if $r {>} 0$, $\evalInState{\varphi(z)}{x_i}$ has a time-derivative of value $\evalInState{\varphi(z)}{\theta_i}$ \emph{at all $z \in [0, r]$}.
\item The value of other variables $y \not \in \{x_1, \dots, x_n\}$ remains constant throughout the continuous evolution, that is $\evalInState{\varphi(z)}{y} = \evalInState{\firstStateSymbol}{y}$ for all times $z \in [0,r]$;
\item $\varphi$ respects the evolution domain at all times: $\varphi(z) \models \evDomainConstraint$ for all $z \in [0, r]$.
\end{enumerate}
\end{definition}

Because we have not generalized the syntax to admit systems of differential equations with non-classical solutions, \rref{def:CaraDef} and \rref{def:classicalDef} are equivalent.
In other words, $\secondStateSymbol$ is reachable from $\firstStateSymbol$ as in \rref{def:classicalDef} iff $\secondStateSymbol$ is reachable in the extended sense from $\firstStateSymbol$ (as in \rref{def:CaraDef}).
In \rref{def:dTLTraceSemantics}, the only change we made from \cite[Definition 4.3]{Pl} was to the traces of ODEs.
Since the traces of ODEs have not actually changed, the semantics of hybrid programs in \PdTL and \dTL is the same.

In \dL, the transition semantics of hybrid programs is given by the following reachability relation $\reach$ \cite{Pl,Pl2}:

\begin{definition}\label{def:reachDef} For a hybrid program $\alpha$, $\reach(\alpha)$ is defined inductively by:
\begin{enumerate}
\item $\reach(x:=e) = \{(\firstStateSymbol, \secondStateSymbol): \secondStateSymbol = \updateState{\firstStateSymbol}{x}{val(\firstStateSymbol, e)} \text{ for } \firstStateSymbol \in \stateSet{\realValuedVariableSet}\}$
\item $\reach(x' = f(x)\ \&\ \evDomainConstraint) = \{(\firstStateSymbol, \secondStateSymbol) \text{ where } \firstStateSymbol \text{ is reachable from } \secondStateSymbol \text{ as in \rref{def:classicalDef}}\}$
\item $\reach(\alpha \cup \beta) = \reach(\alpha) \cup \reach(\beta)$
\item $\reach(?\FOLformulaSymbol) = \{(\firstStateSymbol, \firstStateSymbol) : val(\firstStateSymbol, \FOLformulaSymbol) = true\}$
\item $\reach(\alpha; \beta) = \{(\firstStateSymbol_1, \firstStateSymbol_2)\ |\ (\firstStateSymbol_1, \secondStateSymbol) \in \reach(\alpha), (\secondStateSymbol, \firstStateSymbol_2) \in \reach(\beta)  \text{ for a state } \secondStateSymbol\}$
\item $\tau(\alpha^*) = \bigcup_{n \in \mathbb{N}}\reach(\alpha^n)$, where $\alpha^{n+1}$ is defined as $(\alpha^n; \alpha)$ for $n{\geq}1$, and $\alpha^0$ is defined as $?(true)$
\end{enumerate}
\end{definition}

We have the following correspondence between the reachability relation $\rho$ and our trace semantics $\tau$.
This result corresponds to \cite[Lemma 4.1]{Pl} for \dTL, and the proof is the same, since the proof in \cite{Pl} only requires reasoning about the semantics of hybrid programs and, as discussed above, our semantics of hybrid programs is the same as that of \dTL.

\begin{lemma}\label{lem:reachTrace}
For hybrid programs $\alpha$, we have 
$$\reach(\alpha) = \{(\text{first }\tracesymbol, \text{last }\tracesymbol)\ |\ \tracesymbol \in \tau(\alpha)\text{ terminates}\},$$
where $\reach(\alpha)$ is the reachability relation for hybrid programs.
\end{lemma}

\begin{proof} Identical to \cite[Lemma 4.1]{Pl}.
\end{proof}
We also have the following proposition. The proof is the same as \cite[Proposition 4.1]{Pl}, because it only requires reasoning about state formulas, and the semantics of state formulas is unchanged from \dTL.
\begin{proposition}\label{conserv}
\PdTL is a conservative extension of \dL, i.e. all valid \dL formulas are valid in the logic for \PdTL.
\end{proposition}
\begin{proof}
Identical to \cite[Proposition 4.1]{Pl}.
\end{proof}

Now we show that the \dL axioms in \rref{fig:dLaxioms} are sound for state formulas. Intuitively, this works because the proofs of the \dL axioms do not in any way depend on the actual particulars of the postconditions; only on the semantics of the hybrid programs, which is unchanged for \PdTL.

\begin{proposition} The axioms of \dL in \rref{fig:dLaxioms} (as developed in \cite{Pl2}) are sound for state formulas of \PdTL:
\begin{figure}[tbhp]
  \def\arraystretch{1.4}
    \begin{tabular}{l c} 
   $[:=]\ \  [x := e]\formulaSymbol(x) \leftrightarrow \formulaSymbol(e)$ & 
    \\
    $[?]\ \ [?\secondFOLformulaSymbol]\formulaSymbol \leftrightarrow (\secondFOLformulaSymbol \to \formulaSymbol)$ & \\ 
    $[']\ \ [x' = f(x)]\formulaSymbol(x) \leftrightarrow \forall t{\geq}0 [x := y(t)]\formulaSymbol(x)\ \  (y'(t) = f(y))$ & \\
    $['\&]\ \ [x' = f(x)\&q(x)]\formulaSymbol(x) \leftrightarrow \forall t{\geq}0 ((\forall 0{\leq}s{\leq}t\ q(y(s)) \to [x := y(t)]\formulaSymbol(x))\ \  (y'(t) = f(y))$ 
    \\
    $[\cup]\ \  [\alpha \cup \beta] \formulaSymbol \leftrightarrow [\alpha]\formulaSymbol \land [\beta]\formulaSymbol$ & \\
    $[;]\ \ [\alpha; \beta]\formulaSymbol \leftrightarrow [\alpha][\beta]\formulaSymbol$ & \\
    $[^*]\ \  [\alpha^*]\formulaSymbol \leftrightarrow \formulaSymbol \land [\alpha][\alpha^*]\formulaSymbol$ & \\
    $\langle \cdot \rangle\ \ \langle \alpha \rangle \formulaSymbol \leftrightarrow \lnot [\alpha] \lnot \formulaSymbol$
    \end{tabular}
        \caption{\dL Axioms for State Formulas}
    \label{fig:dLaxioms}
\end{figure}

\end{proposition}

\begin{proof}
Case $[:=]$: By \rref{def:stateformulaSemantics}, $\firstStateSymbol \models [x := e]\formulaSymbol(x)$ iff for every terminating trace $\tracesymbol$ of $x := e$ with $\text{first } \tracesymbol = \firstStateSymbol$, $\val{\tracesymbol}{\formulaSymbol(x)}$ is true. Using \rref{lem:reachTrace} together with \rref{def:traceformulasemantics}, this holds iff for all states $\secondStateSymbol$ with $(\firstStateSymbol, \secondStateSymbol) \in \reach(x := e)$, $\secondStateSymbol \models \formulaSymbol(x)$. By \rref{def:reachDef}, $(\firstStateSymbol, \secondStateSymbol) \in \reach(x := e)$ iff $\secondStateSymbol = \updateState{\firstStateSymbol}{x}{\val{\firstStateSymbol}{e}}$. So $\firstStateSymbol \models [x := e]\formulaSymbol(x)$ iff $\updateState{\firstStateSymbol}{x}{\val{\firstStateSymbol}{e}} \models \formulaSymbol(x)$, which holds iff $\firstStateSymbol \models \formulaSymbol(e)$.
\\ \\
Case $[?]$: By \rref{def:stateformulaSemantics}, $\firstStateSymbol \models [?\secondFOLformulaSymbol]\formulaSymbol$ iff for every terminating trace $\tracesymbol$ of $?\secondFOLformulaSymbol$ with $\text{first } \tracesymbol = \firstStateSymbol$, $\val{\tracesymbol}{\formulaSymbol}$ is true. Using \rref{lem:reachTrace} together with \rref{def:traceformulasemantics}, this holds iff for all states $\secondStateSymbol$ with $(\firstStateSymbol, \secondStateSymbol) \in \reach(?\secondFOLformulaSymbol)$, $\secondStateSymbol \models \formulaSymbol$. By \rref{def:reachDef}, $(\firstStateSymbol, \secondStateSymbol) \in \reach(?\secondFOLformulaSymbol)$ iff $\secondStateSymbol = \firstStateSymbol$ and $\firstStateSymbol \models \secondFOLformulaSymbol$. So  $\firstStateSymbol \models [?\secondFOLformulaSymbol]\formulaSymbol$ iff $\firstStateSymbol \models \secondFOLformulaSymbol \to \formulaSymbol$, and thus the axiom holds.
\\ \\
Case $[']$: Assume that $x' = f(x)$ has the (classical unique global) solution $y(t)$ where $y'(t) = f(y)$ and where $t$ is a fresh variable. By \rref{def:stateformulaSemantics}, $\firstStateSymbol \models [x' = f(x)]\formulaSymbol(x)$ iff for every terminating trace $\tracesymbol$ of $x' = f(x)$ with $\text{first } \tracesymbol = \firstStateSymbol$, $\val{\tracesymbol}{\formulaSymbol(x)}$ is true. Using \rref{lem:reachTrace} together with \rref{def:traceformulasemantics}, this holds iff for all states $\secondStateSymbol$ with $(\firstStateSymbol, \secondStateSymbol) \in \reach(x' = f(x))$, $\secondStateSymbol \models \formulaSymbol(x)$. By \rref{def:reachDef}, $(\firstStateSymbol, \secondStateSymbol) \in \reach(x' = f(x))$ iff there is a function $\varphi: [0, r] \to \stateSet{\realValuedVariableSet}$ satisfying the conditions of \rref{def:classicalDef}.

Since we are assuming that $y(t)$ is a classical unique global solution of $x'{=}f(x)$, we have that for all $r \geq 0$ there is a function $\varphi: [0, r] \to \stateSet{\realValuedVariableSet}$ that satisfies the conditions of \rref{def:classicalDef}, and so $\reach(x'{=}f(x)) = \{(\firstStateSymbol, \varphi(r))\ |\ r \geq 0\}$. Further, we have the following semantic correspondence between $y(t)$ and $\varphi$, namely, $\varphi(r) \models \secondStateSymbol$ iff $\updateState{\firstStateSymbol}{t}{r} \models [x := y(t)]\formulaSymbol(x)$. Thus $\firstStateSymbol \models [x'{=}f(x)]\formulaSymbol(x)$ iff $\firstStateSymbol \models \forall t{\geq}0\ [x:=y(t)]\formulaSymbol(x)$.
\\ \\
Case $['\&]$:  Assume that $x'{=}f(x)$ has the (classical unique global) solution $y(t)$ where $y'(t) = f(y)$ and where $t$ is a fresh variable. By \rref{def:stateformulaSemantics}, $\firstStateSymbol \models [x'{=}f(x)\&q(x)]\formulaSymbol(x)$ iff for every terminating trace $\tracesymbol$ of $x'{=}f(x) \& q(x)$ with $\text{first } \tracesymbol = \firstStateSymbol$, $\val{\tracesymbol}{\formulaSymbol(x)}$ is true. Using \rref{lem:reachTrace} together with \rref{def:traceformulasemantics}, this holds iff for all states $\secondStateSymbol$ with $(\firstStateSymbol, \secondStateSymbol) \in \reach(x'{=}f(x))$, $\secondStateSymbol \models \formulaSymbol(x)$. By \rref{def:reachDef}, $(\firstStateSymbol, \secondStateSymbol) \in \reach(x'{=}f(x))$ iff there is a function $\varphi: [0, r] \to \stateSet{\realValuedVariableSet}$ satisfying the conditions of \rref{def:classicalDef}.

Since we are assuming that $y(t)$ is a classical unique global solution of $x'{=}f(x)$, we have that for all $r \geq 0$ there is a function $\varphi: [0, r] \to \stateSet{\realValuedVariableSet}$ that satisfies the first three conditions of \rref{def:classicalDef}, and so $\reach(x'{=}f(x)) = \{(\firstStateSymbol, \varphi(r))\ |\ r \geq 0, \forall 0{\leq}s{\leq}r\ \ \varphi(s) \models q(x)\}$, where $\forall 0{\leq}s{\leq}r\ \ \varphi(s) \models q(x)$ ensures the fourth condition of \rref{def:classicalDef}. Further, we have the following semantic correspondence between $y(t)$ and $\varphi$, namely, $\varphi(r) \models \secondStateSymbol$ iff $\updateState{\firstStateSymbol}{t}{r} \models [x := y(t)]\formulaSymbol(x)$. Using this semantic correspondence, we achieve $\firstStateSymbol \models [x'{=}f(x)]\formulaSymbol(x)$ iff $\firstStateSymbol \models \forall t{\geq}0 ((\forall 0{\leq}s{\leq}t\ \ q(y(s))) \to [x := y(t)]\formulaSymbol(x))$.
\\ \\
Case $[\cup]$:  By \rref{def:stateformulaSemantics}, $\firstStateSymbol \models [\alpha \cup \beta]\formulaSymbol$ iff for every terminating trace $\tracesymbol$ of $\alpha \cup \beta$ with $\text{first } \tracesymbol = \firstStateSymbol$, $\val{\tracesymbol}{\formulaSymbol}$ is true. Using \rref{lem:reachTrace} together with \rref{def:traceformulasemantics}, this holds iff for all states $\secondStateSymbol$ with $(\firstStateSymbol, \secondStateSymbol) \in \reach(\alpha \cup \beta)$, $\secondStateSymbol \models \formulaSymbol$. By \rref{def:reachDef}, $(\firstStateSymbol, \secondStateSymbol) \in \reach(\alpha \cup \beta)$ iff $(\firstStateSymbol, \secondStateSymbol) \in \reach(\alpha) \cup \reach(\beta)$, and thus $\firstStateSymbol \models [\alpha \cup \beta]\formulaSymbol$ iff for all states $\secondStateSymbol$ with $(\firstStateSymbol, \secondStateSymbol) \in \reach(\alpha)$, $\secondStateSymbol \models \formulaSymbol$ and for all states $\secondStateSymbol$ with $(\firstStateSymbol, \secondStateSymbol) \in \reach(\beta)$, $\secondStateSymbol \models \formulaSymbol$; i.e. $\firstStateSymbol \models [\alpha \cup \beta]\formulaSymbol$ iff $\firstStateSymbol \models [\alpha]\formulaSymbol \land [\beta]\formulaSymbol$.
\\ \\
Case $[;]$: By \rref{def:stateformulaSemantics}, $\firstStateSymbol \models [\alpha;\beta]\formulaSymbol$ iff for every terminating trace $\tracesymbol$ of $\alpha; \beta$ with $\text{first } \tracesymbol = \firstStateSymbol$, $\val{\tracesymbol}{\formulaSymbol}$ is true. Using \rref{lem:reachTrace} together with \rref{def:traceformulasemantics}, this holds iff for all states $\secondStateSymbol$ with $(\firstStateSymbol, \secondStateSymbol) \in \reach(\alpha; \beta)$, $\secondStateSymbol \models \formulaSymbol$. By \rref{def:reachDef}, $(\firstStateSymbol, \secondStateSymbol) \in \reach(\alpha; \beta)$ iff there is some $\secondStateSymbol_1$ with $(\firstStateSymbol, \secondStateSymbol_1) \in \reach(\alpha)$ and $(\secondStateSymbol_1, \secondStateSymbol) \in \reach(\beta)$. 

Then, using \rref{def:stateformulaSemantics}, $\firstStateSymbol \models [\alpha;\beta]\formulaSymbol$ iff for all states $\secondStateSymbol_1$ with $(\firstStateSymbol, \secondStateSymbol_1) \in \reach(\alpha)$, $\secondStateSymbol_1 \models [\beta]\formulaSymbol$. Using \rref{def:stateformulaSemantics}, again, $\firstStateSymbol \models [\alpha;\beta]\formulaSymbol$ iff $\firstStateSymbol \models [\alpha][\beta]\formulaSymbol$.
\\ \\
Case $[^*]$: By \rref{def:stateformulaSemantics}, $\firstStateSymbol \models [\alpha^*]\formulaSymbol$ iff for every terminating trace $\tracesymbol$ of $\alpha^*$ with $\text{first } \tracesymbol = \firstStateSymbol$, $\val{\tracesymbol}{\formulaSymbol}$ is true. Using \rref{lem:reachTrace} together with \rref{def:traceformulasemantics}, this holds iff for all states $\secondStateSymbol$ with $(\firstStateSymbol, \secondStateSymbol) \in \reach(\alpha^*)$, $\secondStateSymbol \models \formulaSymbol$. By \cite{Pl2}, $\reach(\alpha^*) = \reach(?true) \cup \reach(\alpha; \alpha^*)$. So,  $\firstStateSymbol \models [\alpha^*]\formulaSymbol$ iff for all states $\secondStateSymbol$ with $(\firstStateSymbol, \secondStateSymbol) \in \reach(?true)\cup \reach(\alpha;\alpha^*)$, $\secondStateSymbol \models \formulaSymbol$. 

Thus, $\firstStateSymbol \models [\alpha^*]\formulaSymbol$ iff $\firstStateSymbol \models [(?true)]\formulaSymbol \land [\alpha; \alpha^*]\formulaSymbol$, and using axiom $[?]$ and $[;]$ this holds iff $\firstStateSymbol \models (\text{true} \to \formulaSymbol) \land [\alpha][\alpha^*]\formulaSymbol$ or iff $\firstStateSymbol \models \formulaSymbol \land [\alpha][\alpha^*]\formulaSymbol$.
\\ \\
Case $\langle \cdot \rangle$: By \rref{def:stateformulaSemantics}, $\firstStateSymbol \models \langle \alpha \rangle \formulaSymbol$ iff for some terminating trace $\tracesymbol$ of $\alpha^*$ with $\text{first } \tracesymbol = \firstStateSymbol$, $\val{\tracesymbol}{\formulaSymbol}$ is true. Using \rref{lem:reachTrace} together with \rref{def:traceformulasemantics}, this holds iff there exists a $\secondStateSymbol$ with $(\firstStateSymbol, \secondStateSymbol) \in \reach(\alpha)$ and $\secondStateSymbol \models \formulaSymbol$. That holds iff it is not the case that for all states $\secondStateSymbol$ with $(\firstStateSymbol, \secondStateSymbol) \in \reach(\alpha)$,  $\secondStateSymbol \models \lnot \formulaSymbol$, i.e. iff $\firstStateSymbol \not \models [\alpha]\lnot \formulaSymbol$, or $\firstStateSymbol \models \lnot [\alpha]\lnot \formulaSymbol$.
\end{proof}

\section{Soundness of Proof Calculus}\label{app:ProofCalculusSoundness}

\subsection{Soundness of $[\text{?}]_{\text{tae}}$}
\begin{proof}
Fix a state $\firstStateSymbol$ and take a trace $\tracesymbol$ of $?\FOLformulaSymbol$ with $\text{first }\tracesymbol = \firstStateSymbol$.
Using \rref{def:dTLTraceSemantics}, we see that $\tracesymbol$ either has the form $(\hat{\firstStateSymbol}, \Lambda)$ or $(\hat{\omega})$.
As $|\hat{\omega}| = 0$, using \rref{def:traceformulasemantics}, $\tae{\tracesymbol}{\formulaSymbol}$ holds iff $\firstStateSymbol \models \cl{\formulaSymbol}$.
Thus for any trace $\tracesymbol$ of $?\FOLformulaSymbol$ with $\text{first }\tracesymbol = \firstStateSymbol$, $\tae{\tracesymbol}{\formulaSymbol}$ holds iff $\firstStateSymbol \models \cl{\formulaSymbol}$.
Therefore by \rref{def:stateformulaSemantics}, $\firstStateSymbol \models \rae{?\FOLformulaSymbol}{\formulaSymbol} \leftrightarrow \firstStateSymbol \models \cl{\formulaSymbol}$.
\end{proof}

\subsection{Soundness of $[\cup]_{\text{tae}}$}
\begin{proof} Take any state $\firstStateSymbol$ with $\firstStateSymbol \models \rae{\alpha \cup \beta}{\formulaSymbol}$. Then, using \rref{def:stateformulaSemantics}, any trace $\tracesymbol$ of $\alpha \cup \beta$ with $\text{first } \tracesymbol = \firstStateSymbol$ satisfies $\tae{\tracesymbol}{\formulaSymbol}$. By \rref{def:dTLTraceSemantics}, $\tau(\alpha \cup \beta) = \tau(\alpha) \cup \tau(\beta)$. Thus, we have that any trace $\tracesymbol \in \tau(\alpha) \cup \tau(\beta)$ with $\text{first } \tracesymbol = \firstStateSymbol$ satisfies $\tae{\tracesymbol}{\formulaSymbol}$. In other words, each trace $\firstSubtraceSymbol \in \tau(\alpha)$ with $\text{first } \firstSubtraceSymbol = \firstStateSymbol$ satisfies $\tae{\firstSubtraceSymbol}{\formulaSymbol}$ and each trace $\secondSubtraceSymbol \in \tau(\beta)$ with $\text{first } \secondSubtraceSymbol = \firstStateSymbol$ satisfies $\tae{\secondSubtraceSymbol}{\formulaSymbol}$. By \rref{def:stateformulaSemantics}, this means that $\firstStateSymbol \models \rae{\alpha}{\formulaSymbol} \land \rae{\beta}{\formulaSymbol}$.

Conversely, if $\firstStateSymbol \models \rae{\alpha}{\formulaSymbol} \land \rae{\beta}{\formulaSymbol}$, then by \rref{def:stateformulaSemantics}, any trace $\tracesymbol$ of $\alpha$ with $\text{first } \tracesymbol = \firstStateSymbol$ satisfies $\tae{\tracesymbol}{\formulaSymbol}$ and any trace $\tracesymbol$ of $\beta$ with $\text{first } \tracesymbol = \firstStateSymbol$ satisfies $\tae{\tracesymbol}{\formulaSymbol}$. Thus we have that any trace $\tracesymbol \in \tau(\alpha) \cup \tau(\beta)$ with $\text{first } \tracesymbol = \firstStateSymbol$ satisfies $\tae{\tracesymbol}{\formulaSymbol}$. Again using that $\tau(\alpha) \cup \tau(\beta) = \tau(\alpha \cup \beta)$, this means that any trace $\tracesymbol$ of $\alpha \cup \beta$ with $\text{first } \tracesymbol = \firstStateSymbol$ satisfies $\tae{\tracesymbol}{\formulaSymbol}$. Thus $\firstStateSymbol \models \rae{\alpha \cup \beta}{\formulaSymbol}$.
 \end{proof}
 
\subsection{Soundness of $\text{G}_{\text{tae}}$}
\begin{proof} Since $\formulaSymbol$ is valid by assumption and $\formulaSymbol \to \cl{\formulaSymbol}$ is valid by \rref{prop:clProp2}, $\cl{\formulaSymbol}$ is valid. 

Now, take any trace $\tracesymbol = (\tracesymbol_0, \tracesymbol_1, \dots, \tracesymbol_n)$ of $\alpha$. Because $\cl{\formulaSymbol}$ is valid, for any $\tracesymbol_j$ with $r_j = 0$ and $\tracesymbol_j(0) \neq \Lambda$, $\tracesymbol_j(0) \models \cl{\formulaSymbol}$. Next, let $\mathcal{U}$ be the set of positions $(i, \zeta)$ with $\tracesymbol_i(\zeta) \neq \Lambda$ and $\tracesymbol_i(\zeta) \not \models \formulaSymbol$. Because $\formulaSymbol$ is valid by assumption, $\mathcal{U} = \emptyset$, and thus if $f$ is the mapping defined in \rref{def:taef} for $\tracesymbol$, $f(\mathcal{U}) = f(\emptyset)$ has measure zero. Thus $\tae{\tracesymbol}{\formulaSymbol}$, and so $\rae{\alpha}{\formulaSymbol}$ is valid as $\tracesymbol$ was arbitrary.
 \end{proof}
 
 \subsection{Soundness of $K_{\text{tae}}$}
\begin{proof} Assume $\rae{\alpha}(\formulaSymbol \to \secondFormulaSymbol)$ is valid and $\cl{\formulaSymbol} \to \cl{\secondFormulaSymbol}$ is valid. We need to show that in any state $\firstStateSymbol$ where $\rae{\alpha}{\formulaSymbol}$ is true, $\rae{\alpha}{\secondFormulaSymbol}$ is true. So, take a state $\firstStateSymbol$ where $\firstStateSymbol \models \rae{\alpha}{\formulaSymbol}$ and consider a trace $\tracesymbol = (\tracesymbol_0, \tracesymbol_1, \dots, \tracesymbol_n)$ of $\alpha$ where $\text{first } \tracesymbol = \firstStateSymbol$. By Definition 4, $\tae{\tracesymbol}{\formulaSymbol}$. Now, for any $i$ where $|\tracesymbol_i| = 0$ and $\tracesymbol_i(0) \neq \Lambda$, it holds that $\tracesymbol_i(0) \models \cl{\formulaSymbol}$, and thus the premise that $\cl{\formulaSymbol} \to \cl{\secondFormulaSymbol}$ is valid implies that $\tracesymbol_i(0) \models \cl{\secondFormulaSymbol}$ is true.

Now, let $f$ be the mapping as in \rref{def:taef} for $\tracesymbol$. Next, consider the set of positions 
$$\mathcal{U}_1 = \{(i, \zeta)\ |\ \tracesymbol_i(\zeta) \not \models \formulaSymbol \text{ and }  \tracesymbol_i(\zeta) \neq \Lambda\}.$$
 Since by assumption $\firstStateSymbol \models \rae{\alpha}{\formulaSymbol}$ holds, $f(\mathcal{U}_1)$ has Lebesgue measure zero by \rref{def:traceformulasemantics} and \rref{def:stateformulaSemantics}. Also consider the set of positions 
 $$\mathcal{U}_2 = \{(i, \zeta)\ |\ \tracesymbol_i(\zeta) \not \models (\formulaSymbol \to \secondFormulaSymbol) \text{ and }  \tracesymbol_i(\zeta) \neq \Lambda \}.$$
Since $\rae{\alpha}{(\formulaSymbol \to \secondFormulaSymbol)}$ is valid by assumption, $f(\mathcal{U}_2)$ has Lebesgue measure zero by \rref{def:traceformulasemantics} and \rref{def:stateformulaSemantics}. 

Let $\mathcal{U}$ be the following set of positions:
$$\mathcal{U} = \{(i, \zeta)\ |\ \tracesymbol_i(\zeta) \not \models \secondFormulaSymbol \text{ and } \tracesymbol_i(\zeta)\neq \Lambda \}.$$

Now, $\mathcal{U} = \widetilde{\mathcal{U}}_1 \cup \widetilde{\mathcal{U}}_2$, where $\widetilde{\mathcal{U}}_1 = \{(i, \zeta)\ |\  \tracesymbol_i(\zeta)\neq \Lambda \text{ and }\tracesymbol_i(\zeta) \not \models \secondFormulaSymbol \text{ and } \tracesymbol_i(\zeta) \not \models \formulaSymbol\}$ and $ \widetilde{\mathcal{U}}_2 = \{(i,\zeta)\ |\  \tracesymbol_i(\zeta)\neq \Lambda \text{ and } \tracesymbol_i(\zeta) \not \models \secondFormulaSymbol \text{ and } \tracesymbol_i(\zeta) \models \formulaSymbol\}$. Notice that $ \widetilde{\mathcal{U}}_2  = \mathcal{U}_2$ since $\lnot(\formulaSymbol \to \secondFormulaSymbol)$ is logically equivalent to $\formulaSymbol \land \lnot \secondFormulaSymbol$, and also notice that $\widetilde{\mathcal{U}}_1 \subseteq \mathcal{U}_1$. By completeness of Lebesgue measure, $f(\widetilde{\mathcal{U}}_1)$ and $f(\widetilde{\mathcal{U}}_2)$ both have measure zero. Thus $f(\widetilde{\mathcal{U}}_1) \cup f(\widetilde{\mathcal{U}}_2)  = f(\widetilde{\mathcal{U}}_1 \cup \widetilde{\mathcal{U}}_2)$ has measure zero. \end{proof}

\subsection{Soundness of $[:=]_{\text{tae}}$}
\begin{proof} Fix a state $\firstStateSymbol$, and take any trace $\tracesymbol$ of $x := e$ with $\text{first } \tracesymbol = \firstStateSymbol$.
By \rref{def:dTLTraceSemantics}, $\tracesymbol$ has the form $(\hat{\firstStateSymbol}, \hat{\secondStateSymbol}),$ where $\secondStateSymbol = \updateState{\firstStateSymbol}{x}{val(\firstStateSymbol, e)}$.
By \rref{def:stateformulaSemantics}, $\firstStateSymbol \models \rae{x := e}{\formulaSymbol}$ iff $\tae{\tracesymbol}{\formulaSymbol}$.

Again using \rref{def:dTLTraceSemantics}, $\hat{\firstStateSymbol}$ and $\hat{\secondStateSymbol}$ have duration 0 and satisfy $\hat{\firstStateSymbol}(0) = \firstStateSymbol$ and $\hat{\secondStateSymbol}(0) = \secondStateSymbol$, and so by \rref{def:traceformulasemantics}, $\tae{\tracesymbol}{\formulaSymbol}$ iff both $\hat{\firstStateSymbol}(0) \models \cl{\formulaSymbol}$ and $\hat{\secondStateSymbol}(0) \models \cl{\formulaSymbol}$, i.e. iff both $\firstStateSymbol \models \cl{\formulaSymbol}$ and $\secondStateSymbol \models \cl{\formulaSymbol}$ hold. 
Now, $\secondStateSymbol \models \cl{\formulaSymbol}$ iff $\firstStateSymbol \models [x := e] \cl{\formulaSymbol}$, because $\secondStateSymbol$ is the unique successor of $\firstStateSymbol$ under $x:= e$.
Thus for any state $\firstStateSymbol$, $\firstStateSymbol \models  \rae{x := e}{\formulaSymbol}$ iff $\firstStateSymbol \models \cl{\formulaSymbol} \land [x := e] \cl{\formulaSymbol}$. 
 \end{proof}
 
 \subsection{Soundness of $[;]_{\text{tae}}$}
\begin{proof} Fix a state $\firstStateSymbol$ and assume $\firstStateSymbol \models \rae{\alpha}{\formulaSymbol}\ \land\ [\alpha] \rae{\beta}{\formulaSymbol}$. Take a trace $\tracesymbol$ of $\alpha; \beta$ where $\text{first } \tracesymbol = \firstStateSymbol$.  By \rref{def:dTLTraceSemantics}, $\tracesymbol = \firstSubtraceSymbol \circ \secondSubtraceSymbol$, where $\firstSubtraceSymbol$ is a trace of $\alpha$ with $\text{first } \firstSubtraceSymbol = \firstStateSymbol$ and $\secondSubtraceSymbol$ is a trace of $\beta$.  If $\firstSubtraceSymbol$ does not terminate, then $\firstSubtraceSymbol \circ \secondSubtraceSymbol = \firstSubtraceSymbol$ by \rref{def:dTLTraceSemantics}.  By assumption, since $\firstStateSymbol \models \rae{\alpha}{\formulaSymbol}$, $\tae{\firstSubtraceSymbol}{\formulaSymbol}$, and thus $\tae{\firstSubtraceSymbol \circ \secondSubtraceSymbol}{\formulaSymbol}$ in that case.  Else, if $\firstSubtraceSymbol$ terminates, then $\text{last } \firstSubtraceSymbol = \text{first }\secondSubtraceSymbol$. From our assumption that $\firstStateSymbol \models \rae{\alpha}{\formulaSymbol}$, we have that $\tae{\firstSubtraceSymbol}{\formulaSymbol}$, since $\firstSubtraceSymbol \in \tau(\alpha)$. Next, by our assumption that $\firstStateSymbol \models [\alpha] \rae{\beta}{\formulaSymbol}$, we have that $\text{last } \firstSubtraceSymbol \models \rae{\beta}{\formulaSymbol}.$ Then because $\text{last } \firstSubtraceSymbol = \text{first } \secondSubtraceSymbol$, we have that $\tae{\secondSubtraceSymbol}{\formulaSymbol}$. Therefore we have that both $\tae{\firstSubtraceSymbol}{\formulaSymbol}$ and $\tae{\secondSubtraceSymbol}{\formulaSymbol}$, and by \rref{lem:taetracecomp}, we obtain $\tae{\firstSubtraceSymbol \circ \secondSubtraceSymbol}{\formulaSymbol}$. Since $\firstSubtraceSymbol \circ \secondSubtraceSymbol = \tracesymbol$ was an arbitrary trace of $\alpha; \beta$ with $\text{first } \tracesymbol = \firstStateSymbol$, it follows that $\firstStateSymbol \models \rae{\alpha; \beta}{\formulaSymbol}$, as desired.

Conversely, fix a state $\firstStateSymbol$ with $\firstStateSymbol \models \rae{\alpha; \beta}{\formulaSymbol}.$   If $\firstStateSymbol \not \models \rae{\alpha}{\formulaSymbol}$, then there is some trace $\firstSubtraceSymbol$ of $\alpha$ with $\text{first } \firstSubtraceSymbol = \firstStateSymbol$ and $\ntae{\firstSubtraceSymbol}{\formulaSymbol}$.  Since $\firstSubtraceSymbol$ is a prefix of a trace of $\alpha; \beta$ by \rref{def:dTLTraceSemantics} \cite{DBLP:conf/cade/JeanninP14}, this contradicts our assumption that $\firstStateSymbol \models \rae{\alpha; \beta}{\formulaSymbol}$. Thus $\firstStateSymbol \models \rae{\alpha}{\formulaSymbol}$.  

Also, if $\firstStateSymbol \not \models [\alpha]\rae{\beta}{\formulaSymbol}$, then by \rref{def:stateformulaSemantics} there is some terminating trace $\firstSubtraceSymbol$ of $\alpha$ with $\text{first } \firstSubtraceSymbol  = \firstStateSymbol$ where $\val{\firstSubtraceSymbol}{\rae{\beta}{\formulaSymbol}}$ is false. Now, by \rref{def:traceformulasemantics}, that means that $\val{\text{last } \firstSubtraceSymbol}{\rae{\beta}{\formulaSymbol}}$ is false. By \rref{def:stateformulaSemantics}, that means that there is a trace $\secondSubtraceSymbol$ of $\beta$ with $\text{first }\secondSubtraceSymbol = \text{last }\firstSubtraceSymbol$ with $\ntae{\secondSubtraceSymbol}{\formulaSymbol}$. Then by \rref{def:dTLTraceSemantics}, $\firstSubtraceSymbol \circ \secondSubtraceSymbol \in \tau(\alpha; \beta)$. Since $\ntae{\secondSubtraceSymbol}{\formulaSymbol}$, by \rref{lem:taetracecomp} we have $\ntae{\firstSubtraceSymbol \circ \secondSubtraceSymbol}{\formulaSymbol}$, and this contradicts our assumption that $\firstStateSymbol \models \rae{\alpha; \beta}{\formulaSymbol}$. Thus $\firstStateSymbol \models \rae{\beta}{\formulaSymbol}$, and in particular, $\firstStateSymbol \models \rae{\alpha}{\formulaSymbol} \land \rae{\beta}{\formulaSymbol}$, as desired.
 \end{proof}

\subsection{Soundness of $\text{I}_{\text{tae}}$}
\begin{proof} Take a state $\firstStateSymbol$ with $\firstStateSymbol \models \cl{\formulaSymbol} \land [\alpha^*](\mrae{\cl{\formulaSymbol}}{\alpha}{\formulaSymbol})$. We proceed by induction on the number of iterations of $\alpha$. Because $\firstStateSymbol \models \cl{\formulaSymbol}$, $\firstStateSymbol \models \rae{(?true)}{\formulaSymbol}$, and so $\firstStateSymbol \models \rae{\alpha^0}{\formulaSymbol}$. Now, assume that $\firstStateSymbol \models \rae{\alpha^n}{\formulaSymbol}$. Take any trace $\tracesymbol$ of $\alpha^{n+1}$ with $\text{first }\tracesymbol =\firstStateSymbol$. By \rref{def:dTLTraceSemantics}, $\tracesymbol$ is of the form $\firstSubtraceSymbol \circ \secondSubtraceSymbol$ where $\firstSubtraceSymbol$ is a trace of $\alpha^n$ where $\text{first } \firstSubtraceSymbol = \firstStateSymbol$ and $\secondSubtraceSymbol$ is a trace of $\alpha$.

If $\firstSubtraceSymbol$ does not terminate, then since by assumption $\firstStateSymbol \models \rae{\alpha^n}{\formulaSymbol}$, we have that $\tae{\firstSubtraceSymbol}{\formulaSymbol}$ and thus since $\firstSubtraceSymbol \circ \secondSubtraceSymbol = \firstSubtraceSymbol$ by \rref{def:dTLTraceSemantics}, we have $\tae{\firstSubtraceSymbol \circ \secondSubtraceSymbol}{\formulaSymbol}$.

Else, if $\firstSubtraceSymbol$ terminates, then by \rref{def:dTLTraceSemantics}, $\text{last }\firstSubtraceSymbol = \text{first }\secondSubtraceSymbol$. From our assumption that $\firstStateSymbol \models \rae{\alpha^n}{\formulaSymbol}$ and \rref{lem:laststatelemma}, we have that $\text{last } \firstSubtraceSymbol \models \cl{\formulaSymbol}$, and so $\text{first }\secondSubtraceSymbol \models \cl{\formulaSymbol}$.  Now, because $\firstStateSymbol \models [\alpha^*](\mrae{\cl{\formulaSymbol}}{\alpha}{\formulaSymbol})$, in particular $\firstStateSymbol \models [\alpha^n](\mrae{\cl{\formulaSymbol}}{\alpha}{\formulaSymbol})$. Thus, any trace $\Upsilon$ of $\alpha$ that starts in a state $\gamma$ where $\gamma$ is a final state of $\alpha^n$ and $\gamma \models \cl{\formulaSymbol}$ satisfies $\tae{\Upsilon}{\formulaSymbol}$. Thus in particular, $\tae{\secondSubtraceSymbol}{\formulaSymbol}$. By our inductive hypothesis that $\firstStateSymbol \models \rae{\alpha^n}{\formulaSymbol}$, we also have $\tae{\firstSubtraceSymbol}{\formulaSymbol}$, and thus by \rref{lem:taetracecomp}, $\tae{\firstSubtraceSymbol \circ \secondSubtraceSymbol}{\formulaSymbol}$. 

In both cases, $\tae{\firstSubtraceSymbol \circ \secondSubtraceSymbol}{\formulaSymbol}$. Since $\tracesymbol = \firstSubtraceSymbol \circ \secondSubtraceSymbol$ was an arbitrary trace of $\alpha^{n+1}$, it follows that $\firstStateSymbol \models \rae{\alpha^{n+1}}{\formulaSymbol}$. By induction, $\firstStateSymbol \models \rae{\alpha^*}{\formulaSymbol}$, as desired.

Conversely, take a state $\firstStateSymbol$ with $\firstStateSymbol \models \rae{\alpha^*}{\formulaSymbol}$. Then $\firstStateSymbol \models \rae{\alpha^0}{\formulaSymbol}$ by assumption. Now, $\tau(\alpha^0) = \tau(?true)$ by \rref{def:dTLTraceSemantics}, and so each trace $\tracesymbol$ in $\tau(\alpha^0)$ with $\text{first } \tracesymbol = \firstStateSymbol$ is of the form $(\tracesymbol_0)$ where $|\tracesymbol_0| = 0$, and so $\text{first } \tracesymbol = \text{last } \tracesymbol = \firstStateSymbol$. Using \rref{lem:laststatelemma}, $\text{last }\tracesymbol \models \cl{\formulaSymbol}$, and so $\firstStateSymbol \models \cl{\formulaSymbol}$.

Next, we wish to show that $\firstStateSymbol \models [\alpha^*](\mrae{\cl{\formulaSymbol}}{\alpha}{\formulaSymbol})$. This is equivalent to showing that $\firstStateSymbol \models [\alpha^n](\mrae{\cl{\formulaSymbol}}{\alpha}{\formulaSymbol})$ for each $n{\geq}0$.
By \rref{def:stateformulaSemantics}, we must show that for all terminating traces $\tracesymbol$ of $\alpha^n$ with $\text{first } \tracesymbol = \firstStateSymbol$, $\text{last } \tracesymbol \models \mrae{\cl{\formulaSymbol}}{\alpha}{\formulaSymbol}$.
Because $\firstStateSymbol \models \rae{\alpha^*}{\formulaSymbol}$ by assumption, we know $\firstStateSymbol \models \rae{\alpha^{n+1}}{\formulaSymbol}$.
Now, for any terminating trace $\tracesymbol$ of $\alpha^n$ with $\text{first } \tracesymbol = \firstStateSymbol$, if $\text{last }\tracesymbol \not \models \rae{\alpha}{\formulaSymbol}$, then there is some trace $\firstSubtraceSymbol$ of $\alpha$ with $\text{first } \firstSubtraceSymbol = \text{last } \tracesymbol$ and $\ntae{\firstSubtraceSymbol}{\formulaSymbol}$, and thus $\ntae{\tracesymbol \circ \firstSubtraceSymbol}{\formulaSymbol}$ by \rref{lem:taetracecomp}.
But because $\tracesymbol \circ \firstSubtraceSymbol$ is a trace of $\alpha^{n+1}$ with $\text{first }\tracesymbol \circ \firstSubtraceSymbol = \firstStateSymbol$, this contradicts that $\firstStateSymbol \models \rae{\alpha^{n+1}}{\formulaSymbol}$.

Therefore, for any terminating trace $\tracesymbol$ of $\alpha^n$ with $\text{first } \tracesymbol = \firstStateSymbol$, $\text{last }\tracesymbol \models \rae{\alpha}{\formulaSymbol}$. Thus $\text{last }\tracesymbol \models (\cl{\formulaSymbol} \to \rae{\alpha}{\formulaSymbol})$, and so $\firstStateSymbol \models [\alpha^n](\mrae{\cl{\formulaSymbol}}{\alpha}{\formulaSymbol})$, as desired.
 \end{proof}

\subsection{Proof of \rref{prop:definable}}\label{app:propositionDefinable}
\begin{proof}
Case: $\FOLformulaSymbol$ is $e = 0$. Define $Q = g(\FOLformulaSymbol)$ to be $e = 0$. 
\begin{enumerate}
    \item If for some $k{\geq}0$, $\updateState{\firstStateSymbol}{t}{k} \not \models \secondFOLformulaSymbol$, then $\updateState{\firstStateSymbol}{t}{k} \not \models e = 0$, so $\updateState{\firstStateSymbol}{t}{k} \models e {>} 0$ or $\updateState{\firstStateSymbol}{t}{k} \models e{<}0$. Either way, since $e$ is a polynomial (and thus continuous in $t$), there exists $\ell {>} k$ so that $\updateState{\firstStateSymbol}{t}{q} \not \models e = 0$ holds for all $q \in [k, \ell)$. Further, if $k {>} 0$, then there exist $\ell_1, \ell_2$ with $\ell_1{<}k{<}\ell_2$ so that $\updateState{\firstStateSymbol}{t}{q} \not \models e = 0$ holds for all $q \in (\ell_1, \ell_2)$.
    \item Since $\secondFOLformulaSymbol$ and $\FOLformulaSymbol$ are the same formula, there is no $k{\geq}0$ where $\updateState{\firstStateSymbol}{t}{k} \models \secondFOLformulaSymbol \land \lnot \FOLformulaSymbol$.
\end{enumerate}
Case: $\FOLformulaSymbol$ is $e{\geq}0$. Define $\secondFOLformulaSymbol = g(\FOLformulaSymbol)$ to be $ e{\geq}0$.
\begin{enumerate}
    \item If for some $k{\geq}0$, $\updateState{\firstStateSymbol}{t}{q} \not \models \secondFOLformulaSymbol$, then $\updateState{\firstStateSymbol}{t}{q} \models e{<}0$. But since $e$ is a polynomial, $e$ is continuous in $t$, and so there exists $\ell {>} k$ such that $\updateState{\firstStateSymbol}{t}{q} \not \models e{\geq}0$ holds for all $q \in [k, \ell)$. Further, if $k {>} 0$, then there exist $\ell_1, \ell_2$ with $\ell_1{<}k{<}\ell_2$ so that $\updateState{\firstStateSymbol}{t}{q} \not \models e{\geq}0$ holds for all $q \in (\ell_1, \ell_2)$.
    \item Since $\secondFOLformulaSymbol$ and $\FOLformulaSymbol$ are the same formula, there is no $k{\geq}0$ where $\updateState{\firstStateSymbol}{t}{k} \models \secondFOLformulaSymbol \land \lnot \FOLformulaSymbol$.
\end{enumerate}
Case: $\FOLformulaSymbol$ is $e{<}0$. We may regard $e$ as a polynomial in $t$---i.e., there is some polynomial $a_nt^n + \cdots + a_1t + a_0$, where $a_n, \dots, a_0$ are FOL terms possibly involving variables other than $t$, such that $\firstStateSymbol \models e = a_nt^n + \cdots + a_1t + a_0$ for all states $\firstStateSymbol$.  Now, let $\secondFOLformulaSymbol = g(\FOLformulaSymbol)$ be defined as follows:
\begin{equation}
 \secondFOLformulaSymbol \equiv e{\leq}0 \land ((a_n = 0 \land \cdots \land a_{1} = 0) \to e{<}0).
\end{equation}
\begin{enumerate}
    \item If for some $k{\geq}0$, $\updateState{\firstStateSymbol}{t}{k} \not \models \secondFOLformulaSymbol$, then either $\updateState{\firstStateSymbol}{t}{k} \models (a_n = 0 \land \cdots \land a_{1} = 0)$ and $\updateState{\firstStateSymbol}{t}{k} \not \models e{<}0$ or $\updateState{\firstStateSymbol}{t}{k} \not \models e{\leq}0$. In the first case, since $a_1, \dots, a_n$ do not depend on $t$, $\updateState{\firstStateSymbol}{t}{q} \models (a_n = 0 \land \cdots \land a_1 = 0)$ holds for all $q{\geq}0$. This means that $\updateState{\firstStateSymbol}{t}{q} \models e = a_0$ holds for all $q{\geq}0$. Because $a_0$ does not depend on $t$, if $\updateState{\firstStateSymbol}{t}{k} \not \models e{<}0$, then $\updateState{\firstStateSymbol}{t}{k} \not \models a_0{<}0$, which implies that $\updateState{\firstStateSymbol}{t}{q} \not \models a_0{<}0$ holds for all $q{\geq}0$. It follows that $\updateState{\firstStateSymbol}{t}{q} \not \models \FOLformulaSymbol$ holds for all $q{\geq}0$.
    
    In the second case, since $\updateState{\firstStateSymbol}{t}{k} \not \models e{\leq}0$, $\updateState{\firstStateSymbol}{t}{k} \models e {>} 0$, and since $e$ is a polynomial, $e$ is continuous in $t$. From this continuity, we have that there is a half-open interval $[k, \ell)$ so that for all $q \in [k, \ell)$, $\updateState{\firstStateSymbol}{t}{q} \models e {>} 0$. Thus we see that for all $q \in [k, \ell)$, $\updateState{\firstStateSymbol}{t}{q} \not \models \FOLformulaSymbol$ holds. Further, if $k {>} 0$, then there is an interval $(\ell_1, \ell_2)$ with $\ell_1{<}k{<}\ell_2$ such that for all $q \in (\ell_1, \ell_2)$, $\updateState{\firstStateSymbol}{t}{q} \not \models \FOLformulaSymbol$ holds.
    \item If $\updateState{\firstStateSymbol}{t}{k} \models \secondFOLformulaSymbol \land \lnot \FOLformulaSymbol$, then we must have that $\updateState{\firstStateSymbol}{t}{k} \models e{\leq}0$ and also $\updateState{\firstStateSymbol}{t}{k} \not \models e{<}0$. 
    From these, it must hold that $\updateState{\firstStateSymbol}{t}{k} \models e = 0$.
    Also, we need $\updateState{\firstStateSymbol}{t}{k} \models (a_n = 0\land \cdots \land a_1 = 0) \to e{<}0$, so it must hold that $\updateState{\firstStateSymbol}{t}{k} \not \models (a_n = 0 \land \cdots \land a_1 = 0)$.
    From this, $\updateState{\firstStateSymbol}{t}{k} \models a_d \neq 0 \land e = a_dt^d + \cdots + a_0$ for some $1{\leq}d{\leq}n$, and so in order for $\updateState{\firstStateSymbol}{t}{k} \models a_d \neq 0 \land a_dt^d + \cdots + a_0 = 0$ to be true, $k$ must be one of the roots of $\evalInState{\omega}{a_d}t^d + \cdots + \evalInState{\omega}{a_0}$.
    Thus there are at most $d$ (and so only finitely many) values of $k$ where $\updateState{\firstStateSymbol}{t}{k} \models \secondFOLformulaSymbol \land \lnot \FOLformulaSymbol$.
\end{enumerate}
Case: $\FOLformulaSymbol$ is $\FOLformulaSymbol_1 \land \FOLformulaSymbol_2$.
Define $\secondFOLformulaSymbol = g(\FOLformulaSymbol)$ to be $g(\FOLformulaSymbol_1) \land g(\FOLformulaSymbol_2)$, and denote $g(\FOLformulaSymbol_1) = \secondFOLformulaSymbol_1$, $g(\FOLformulaSymbol_2) = \secondFOLformulaSymbol_2$.
\begin{enumerate}
    \item If for some $k{\geq}0$, $\updateState{\firstStateSymbol}{t}{k} \not \models \secondFOLformulaSymbol$, then $\updateState{\firstStateSymbol}{t}{k} \not \models \secondFOLformulaSymbol_1$ or $\updateState{\firstStateSymbol}{t}{k} \not \models \secondFOLformulaSymbol_2$.
    In the first case, by the induction hypothesis there is an interval $[k, \hat{\ell})$ with $\hat{\ell} {>} k$ so that for any $q \in [k, \hat{\ell})$, $\updateState{\firstStateSymbol}{t}{q} \not \models \FOLformulaSymbol_1$.
    Further, if $k{>}0$, then by the induction hypothesis there is an interval $(\hat{\ell}_1, \hat{\ell}_2)$ with $\hat{\ell}_1{<}k{<}\hat{\ell}_2$ and for any $q \in (\hat{\ell}_1, \hat{\ell}_2)$, $\updateState{\firstStateSymbol}{t}{q} \not \models \FOLformulaSymbol_1$.
    Because $\FOLformulaSymbol$ is $\FOLformulaSymbol_1 \land \FOLformulaSymbol_2$, this means that for any $q \in [k, \hat{\ell})$, $\updateState{\firstStateSymbol}{t}{q} \not \models \FOLformulaSymbol$ holds; and if $k{>}0$, then for any $q \in (\hat{\ell}_1, \hat{\ell}_2)$, $\updateState{\firstStateSymbol}{t}{q} \not \models \FOLformulaSymbol$.
    
    Similarly, in the second case, by the induction hypothesis there is an interval $[k, \widetilde{\ell})$ with $\widetilde{\ell} {>} k$ so that for any $q \in [k, \widetilde{\ell})$, $\updateState{\firstStateSymbol}{t}{q} \not \models \FOLformulaSymbol_2$.
    Further, if $k{>}0$ then by the induction hypothesis there is an interval $(\widetilde{\ell}_1, \widetilde{\ell}_2)$ with $\widetilde{\ell}_1{<}k{<}\widetilde{\ell}_2$ so that $\updateState{\firstStateSymbol}{t}{q} \not \models \FOLformulaSymbol_2$ for any $q \in (\widetilde{\ell}_1, \widetilde{\ell}_2)$.
    Because $\FOLformulaSymbol$ is $\FOLformulaSymbol_1 \land \FOLformulaSymbol_2$, this means that for any $q \in [k, \widetilde{\ell})$, $\updateState{\firstStateSymbol}{t}{q} \not \models \FOLformulaSymbol$ holds, and if $k{>}0$ then $\updateState{\firstStateSymbol}{t}{q} \not \models \FOLformulaSymbol$ for any $q \in (\widetilde{\ell}_1, \widetilde{\ell}_2)$.
    \item By the induction hypothesis, $\updateState{\firstStateSymbol}{t}{k} \models \secondFOLformulaSymbol_1 \land \lnot \FOLformulaSymbol_1$ only for finitely many $k$, say for $k \in \{u_1, \dots, u_r\}$, and $\updateState{\firstStateSymbol}{t}{k} \models \secondFOLformulaSymbol_2 \land \lnot \FOLformulaSymbol_2$ only for finitely many $k$, say for $k \in \{w_1, \dots, w_s\}$.
    Now, if we have $\updateState{\firstStateSymbol}{t}{k} \models \secondFOLformulaSymbol \land \lnot (\FOLformulaSymbol)$, then we have $\updateState{\firstStateSymbol}{t}{k} \models (\secondFOLformulaSymbol_1 \land \secondFOLformulaSymbol_2) \land \lnot (\FOLformulaSymbol_1 \land \FOLformulaSymbol_2)$, or $\updateState{\firstStateSymbol}{t}{k} \models (\secondFOLformulaSymbol_1 \land \secondFOLformulaSymbol_2) \land (\lnot \FOLformulaSymbol_1 \lor \lnot \FOLformulaSymbol_2)$, which can only happen for $k \in \{u_1, \dots, u_r, w_1, \dots, w_s\}$, i.e. for finitely many $k$.
\end{enumerate}
Case: $\FOLformulaSymbol$ is $\FOLformulaSymbol_1 \lor \FOLformulaSymbol_2$.
Define $\secondFOLformulaSymbol = g(\FOLformulaSymbol)$ to be $\secondFOLformulaSymbol_1 \lor \secondFOLformulaSymbol_2$, and denote $g(\FOLformulaSymbol_1) = \secondFOLformulaSymbol_1$, $g(\FOLformulaSymbol_2) = \secondFOLformulaSymbol_2$.
\begin{enumerate} 
    \item If for some $k{\geq}0$, $\updateState{\firstStateSymbol}{t}{k} \not \models \secondFOLformulaSymbol$, then $\updateState{\firstStateSymbol}{t}{k} \not \models \secondFOLformulaSymbol_1$ and $\updateState{\firstStateSymbol}{t}{k} \not \models \secondFOLformulaSymbol_2$.
    By the induction hypothesis, then, there are intervals $[k, \hat{\ell})$ and $[k, \widetilde{\ell})$ with $\hat{\ell} {>} k$, $\widetilde{\ell} {>} k$ so that for any $q \in [k, \hat{\ell})$, $\updateState{\firstStateSymbol}{t}{q} \not \models \FOLformulaSymbol_1$ and for any $q \in [k, \widetilde{\ell})$, $\updateState{\firstStateSymbol}{t}{q} \not \models \FOLformulaSymbol_2$. Letting $\ell = \text{min}\{\hat{\ell}, \widetilde{\ell}\}$, we see that $\ell {>} k$ and $\updateState{\firstStateSymbol}{t}{q} \not \models \FOLformulaSymbol_1 \lor \FOLformulaSymbol_2$ holds for all $q \in [k, \ell)$. 
    
    Further, if $k > 0$, then by the induction hypothesis there are intervals $(\hat{\ell}_1, \hat{\ell}_2$ and $(\widetilde{\ell}_1, \widetilde{\ell}_2$ with $\hat{\ell}_1{<}k{<}\hat{\ell}_2)$ and  $\widetilde{\ell}_1{<}k{<}\widetilde{\ell}_2)$ so that $\updateState{\firstStateSymbol}{t}{q} \not \models \FOLformulaSymbol_1$ for any $q \in (\hat{\ell}_1, \hat{\ell}_2)$ and $\updateState{\firstStateSymbol}{t}{q} \not \models \FOLformulaSymbol_2$ for any $q \in (\widetilde{\ell}_1, \widetilde{\ell}_2)$. Letting $\ell_1 = \text{max}\{\hat{\ell}_1, \widetilde{\ell}_1\}$ and $\ell_2 = \text{min}\{\hat{\ell}_2, \widetilde{\ell}_2\}$, we see that $\ell_1{<}k{<}\ell_2$ and $\updateState{\firstStateSymbol}{t}{q} \not \models \FOLformulaSymbol_1 \lor \FOLformulaSymbol_2$ holds for all $q \in (\ell_1, \ell_2)$. 
    \item By the induction hypothesis, $\updateState{\firstStateSymbol}{t}{k} \models \secondFOLformulaSymbol_1 \land \lnot \FOLformulaSymbol_1$ only for finitely many values of $k$, say for $k \in \{u_1, \dots, u_r\}$ and $\updateState{\firstStateSymbol}{t}{k} \models \secondFOLformulaSymbol_2 \land \lnot \FOLformulaSymbol_2$ only for finitely many values of $k$, say for $k \in \{w_1, \dots, w_s\}$. Now, if we have $\updateState{\firstStateSymbol}{t}{k} \models \secondFOLformulaSymbol \land \lnot (\FOLformulaSymbol)$, then we have $\updateState{\firstStateSymbol}{t}{k} \models (\secondFOLformulaSymbol_1 \lor \secondFOLformulaSymbol_2) \land \lnot (\FOLformulaSymbol_1 \lor \FOLformulaSymbol_2)$, so $\updateState{\firstStateSymbol}{t}{k} \models (\secondFOLformulaSymbol_1 \lor \secondFOLformulaSymbol_2) \land (\lnot \FOLformulaSymbol_1 \land \lnot \FOLformulaSymbol_2)$, which can only happen for $k \in \{u_1, \dots, u_r, w_1, \dots, w_s\}$, i.e. for finitely many $k$.
\end{enumerate} \end{proof}

\subsection{Soundness of $[']_{\text{tae}}$}
\begin{proof}
Fix a start state $\firstStateSymbol$. By \rref{prop:definable}, the following conditions hold:
\begin{enumerate}
    \item Locally false: If $\updateState{\firstStateSymbol}{t}{k} \not \models \secondFOLformulaSymbol$ for some $k{\geq}0$, then there is a nonempty interval $[k, \ell)$ so that for all $q \in [k, \ell)$, $\updateState{\firstStateSymbol}{t}{q} \not \models [x :=  y(t)] \FOLformulaSymbol$.
    \item Finite difference: There are only finitely many values $k{\geq}0$ where $\updateState{\firstStateSymbol}{t}{k} \models \secondFOLformulaSymbol \land \lnot [x := y(t)] \FOLformulaSymbol$.
\end{enumerate}

From \rref{def:dTLTraceSemantics}, each trace of $x' = f(x)$ that starts at state $\firstStateSymbol$ is of the form $(\varphi)$ where $\varphi$ is a Carath\'{e}odory solution of $x' = f(x)$ of some nonnegative duration $r$ with $\varphi(0) = \firstStateSymbol$.
For clarity, we write $\varphi_r$ to indicate that the solution is of duration $r{\geq}0$.

Recall that we are assuming that $x' = f(x)$ has a unique polynomial global solution $y(t)$, and so each $\varphi_r$ follows $y(t)$ for time $r$.
More precisely, we have the following semantic correspondence between $\varphi_r$ and the polynomial solution $y$, namely that $\varphi_r(k) \models \FOLformulaSymbol$ for some $k$ with $0{\leq}k{\leq}r$ iff $\updateState{\varphi_r(0)}{t}{k} \models [x := y(t)] \FOLformulaSymbol$, i.e. iff $\updateState{\firstStateSymbol}{t}{k}\models [x := y(t)] \FOLformulaSymbol$.

Now, $\firstStateSymbol \models \rae{x' = f(x)}{\FOLformulaSymbol}$ iff the two conditions of \rref{def:traceformulasemantics} hold for all $\varphi_r$.

\begin{enumerate}
\item Case ($r = 0$): The discrete condition of \rref{def:traceformulasemantics} demands $\varphi_0(0) \models \cl{\FOLformulaSymbol}$, and since $\varphi_0(0) = \firstStateSymbol$, this holds iff $\firstStateSymbol \models \cl{\FOLformulaSymbol}$.

\item Case ($r {>} 0$): For every $r {>} 0$, the set of positions where $\varphi_r(\zeta) \not \models P$, $0{\leq}\zeta{\leq}r$ must have Lebesgue measure zero.
Equivalently, this means that for almost all $0{\leq}k{\leq}r$, $\updateState{\varphi_r(0)}{t}{k}\models [x := y(t)] \FOLformulaSymbol$.
If $\firstStateSymbol \models \forall t{\geq}0~Q$, then by the finite difference condition, there are only finitely many $k{\geq}0$ so that $\updateState{\firstStateSymbol}{t}{k}\not\models [x := y(t)] \FOLformulaSymbol$, i.e., for all $r$, it holds that for almost all $0{\leq}k{\leq}r$, $\updateState{\varphi_r(0)}{t}{k}\models [x := y(t)] \FOLformulaSymbol$.

If, however,  $\firstStateSymbol \not\models \forall t{\geq}0~Q$, then there is some $k$ where $\updateState{\firstStateSymbol}{t}{k} \not \models Q$.
Then, by the locally false condition, there is an interval $[k, \ell)$ with $\ell {>} k$ so that for all $q \in [k, \ell), \firstStateSymbol \not \models [x := y(t)] \FOLformulaSymbol$.
Since $[k, \ell)$ has measure $\ell{-}k {>} 0$, for any $r {>} \ell$, it does not hold that the set of positions where $\varphi_r(\zeta) \not \models [x := y(t)] P$ has Lebesgue measure zero.
Thus, the continuous condition of \rref{def:traceformulasemantics} holds iff $\firstStateSymbol \models \forall t{\geq}0~\secondFOLformulaSymbol$.
\end{enumerate}

\noindent So, for any state $\firstStateSymbol$, $\firstStateSymbol \models \rae{x' = f(x)}{\FOLformulaSymbol}$ holds iff $\firstStateSymbol \models \cl{\FOLformulaSymbol} \land \forall t{\geq}0~\secondFOLformulaSymbol$.
Therefore $\rae{x' = f(x)}{\FOLformulaSymbol} \leftrightarrow \cl{\FOLformulaSymbol} \land \forall t{\geq}0 \secondFOLformulaSymbol$ is valid.
 \end{proof}

\subsection{Soundness of $['\&]_{\text{tae}}$}
\begin{proof} The proof is very similar to that of the solution axiom $[']_{\text{tae}}$ without evolution domain constraints. The key difference is that as per \rref{def:CaraDef}, our trace semantics requires solutions to satisfy the evolution domain constraint at all points in time.

Fix a start state $\firstStateSymbol$.
By \rref{def:dTLTraceSemantics}, if $\firstStateSymbol \not \models \evDomainConstraint$, then there is only one trace $\tracesymbol$ of $x'=f(x) \& \evDomainConstraint$ with $\text{first } \tracesymbol = \firstStateSymbol$, and it has the form $\tracesymbol = (\hat{\firstStateSymbol}, \Lambda)$.
In this case, to satisfy the discrete condition of \rref{def:traceformulasemantics}, we must have $\firstStateSymbol \models \cl{\FOLformulaSymbol}$.
Since $\firstStateSymbol \not \models \evDomainConstraint$, and since $y(0) = x$, $\firstStateSymbol \not \models [x := y(0)] \evDomainConstraint$.
So, we have that for each $t{\geq}0$, $\firstStateSymbol \not \models (\forall 0{\leq}s{\leq}t [x := y(s)] \evDomainConstraint)$, and therefore $\firstStateSymbol \models \left(\forall t {\geq}0 \left( \forall 0{\leq}s{\leq}t \ [x := y(s)]\evDomainConstraint \right) \to \secondFOLformulaSymbol \right)$ holds vacuously. 

Otherwise, if $\firstStateSymbol \models \evDomainConstraint$, then all of the traces of $x' = f(x)\& \evDomainConstraint$ are of the form $(\varphi)$ where $\varphi$ is a Carath\'{e}odory solution of $x' = f(x) \& \evDomainConstraint$ of some nonnegative duration $r$, with $\varphi(0) = \firstStateSymbol$.
For clarity, we write $\varphi_r$ to indicate that the solution has duration $r{\geq}0$.
By the semantics of ODEs, (\rref{def:CaraDef}), $\varphi_r(t) \models \evDomainConstraint$ for all $0{\leq}t{\leq}r$.

Recall that we are assuming that $x' = f(x)$ has a unique polynomial global solution $y(t)$, which leads to the following semantic correspondence between $\varphi_r$ and the polynomial solution $y$, namely that $\varphi_r(k) \models \FOLformulaSymbol$ for some $k$ with $0{\leq}k{\leq}r$ iff $\updateState{\varphi_r(0)}{t}{k} \models [x := y(t)] \FOLformulaSymbol$.

In particular, for each trace $\varphi_r$ of $x' = f(x) \& \evDomainConstraint$, since $\varphi_r(t) \models \evDomainConstraint$ for all $0{\leq}t{\leq}r$, we have that for all $0{\leq}s{\leq}r$, $\updateState{\omega}{t}{s} \models [x := y(t)] \evDomainConstraint$, and so $\omega \models \forall 0{\leq}s{\leq}r [x := y(s)] \evDomainConstraint.$
Also, if $\omega \models \forall 0{\leq}s{\leq}r [x := y(s)] \evDomainConstraint,$ then $\varphi_r$ is a trace of $x' = f(x) \& \evDomainConstraint$.
This is to say exists a trace of duration $r \geq 0$ of $x' = f(x) \& \evDomainConstraint$ that starts in $\firstStateSymbol$ iff $\omega \models (\forall 0{\leq}s{\leq}r [x := y(s)] \evDomainConstraint).$

Now, $\firstStateSymbol \models \rae{x' = f(x)}{\FOLformulaSymbol}$ iff the two conditions of \rref{def:traceformulasemantics} hold.

\begin{enumerate}
\item Case $(r = 0)$: Because we know $\omega \models R$, we know that there is a trace of $x' = f(x) \& \evDomainConstraint$ of duration $0$. Then in order to satisfy the discrete condition, $\varphi_0(0) \models \cl{\FOLformulaSymbol}$ must hold, and since $\varphi_0(0) = \firstStateSymbol$, this holds iff $\firstStateSymbol \models \cl{\FOLformulaSymbol}$.

\item Case $(r > 0)$: Second, as long as $r > 0$ is such that $\omega \models (\forall 0{\leq}s{\leq}r [x := y(s)] \evDomainConstraint),$ the ODE $x' = f(x) \& \evDomainConstraint$ has a trace of duration $r$, and then the set of positions where $\varphi_r(\zeta) \not \models P$, $0{\leq}\zeta{\leq}r$ must have Lebesgue measure zero. Fix such an $r$; call it $\hat{r}$.

Then we need that for almost all $0{\leq}k{\leq}\hat{r}$, $\updateState{\varphi_{\hat{r}}(0)}{t}{k}\models [x := y(t)] \FOLformulaSymbol$.
If for all $0{\leq}k{\leq}\hat{r}$, $\updateState{\firstStateSymbol}{t}{k} \models \secondFOLformulaSymbol$, then by the finite difference condition, there are only finitely many $0{\leq}q{\leq}\hat{r}$ where $\updateState{\firstStateSymbol}{t}{q}\not\models [x := y(t)] \FOLformulaSymbol$, so in particular it holds that for almost all $0{\leq}k{\leq}\hat{r}$, $\updateState{\varphi_{\hat{r}}(0)}{t}{k}\models [x := y(t)] \FOLformulaSymbol$.

If, however, for some $0{\leq}q{\leq}\hat{r}$, $\updateState{\firstStateSymbol}{t}{q} \not \models \secondFOLformulaSymbol$, then there are two cases: $q = 0$ or $q > 0$. If $q = 0$, then using the locally false condition, there is an interval $[q, \ell)$ with $q{<}\ell{<}\hat{r}$ so that for all $k \in [q, \ell), \updateState{\firstStateSymbol}{t}{k} \not \models [x := y(t)] \FOLformulaSymbol$, and because $[q, \ell)$ has measure $\ell - q > 0$, it does not hold that for almost all $0{\leq}k{\leq}\hat{r}$, $\updateState{\varphi_{\hat{r}}(0)}{t}{k}\models [x := y(t)] \FOLformulaSymbol$.

If instead $q > 0$, then by the locally false condition there is an interval $(\ell_1, \ell_2)$ with $\ell_1{<}q{<}\ell_2$ so that for all $k \in (\ell_1, \ell_2), \updateState{\firstStateSymbol}{t}{k} \not \models [x := y(t)] \FOLformulaSymbol$, and thus in particular for all $k \in (\ell_1, q]$, $ \updateState{\firstStateSymbol}{t}{k} \not \models [x := y(t)] \FOLformulaSymbol$. Because $(\ell_1, q]$ has measure $q - \ell_1 > 0$, it does not hold that for almost all $0{\leq}k{\leq}\hat{r}$, $\updateState{\varphi_{\hat{r}}(0)}{t}{k}\models [x := y(t)] \FOLformulaSymbol$.

Thus, to satisfy the continuous condition, for every $r$ where  $\omega \models (\forall 0{\leq}s{\leq}r [x := y(s)] \evDomainConstraint)$, we must have $\firstStateSymbol \models \forall 0{\leq}t{\leq}r\ Q$.

Equivalently, we can write this as $\omega \models \forall t {\geq}0\left( \left( \forall 0{\leq}s{\leq}t \ [x := y(s)]\evDomainConstraint \right) \to \secondFOLformulaSymbol \right)$.
\end{enumerate}

So, for any state $\firstStateSymbol$, $\firstStateSymbol \models \rae{x' = f(x)\& \evDomainConstraint}{\FOLformulaSymbol}$ holds iff $$\firstStateSymbol \models \cl{\FOLformulaSymbol} \land \forall t {>}0 \left(\left( \forall 0{\leq}s{\leq}t \ [x := y(s)]\evDomainConstraint \right) \to \secondFOLformulaSymbol \right).$$
Therefore $\rae{x' = f(x)}{\FOLformulaSymbol} \leftrightarrow \cl{\FOLformulaSymbol} \land \forall t {>}0\left( \left( \forall 0{\leq}s{\leq}t \ [x := y(s)]\evDomainConstraint \right) \to \secondFOLformulaSymbol \right)$ is valid.

 \end{proof}
 
\section{Soundness of Derived Rules}\label{app:DerivedSoundness}
\subsection{Inherited Structural Properties}
The following structural properties will be useful in our derivations.
\begin{proposition} \label{prop:structuralprop}
All proof rules of propositional sequent calculus are sound, including the following (where P, Q, and C are \PdTL formulas):
\begin{center}
\begin{minipage}[t]{0.45\textwidth}
\begin{prooftree}
\hypo{\Gamma \vdash P, \Delta}
\hypo{\Gamma, Q \vdash \Delta}
\infer2[$\to$L]{\Gamma, P \to Q \vdash \Delta}
\end{prooftree}
\end{minipage}
\begin{minipage}[t]{0.45\textwidth}
\begin{prooftree}
\hypo{\Gamma, P \vdash Q, \Delta}
\infer1[$\to$R]{\Gamma \vdash P \to Q, \Delta}
\end{prooftree}
\end{minipage}
\end{center}
\begin{center}
\begin{minipage}[t]{0.45\textwidth}
\begin{prooftree}
\hypo{\Gamma \vdash P, \Delta}
\hypo{\Gamma \vdash Q, \Delta}
\infer2[$\land$R]{\Gamma \vdash P \land Q, \Delta}
\end{prooftree}
\end{minipage}
\begin{minipage}[t]{0.45\textwidth}
\begin{prooftree}
\hypo{\Gamma \vdash C, \Delta}
\hypo{\Gamma, C \vdash \Delta}
\infer2[cut]{\Gamma \vdash \Delta}
\end{prooftree}
\end{minipage}
\end{center}
\begin{center}
\begin{minipage}[t]{0.45\textwidth}
\begin{prooftree}
\hypo{\Gamma \vdash \Delta}
\infer1[WR]{\Gamma \vdash P, \Delta}
\end{prooftree}
\end{minipage}
\begin{minipage}[t]{0.45\textwidth}
\begin{prooftree}
\hypo{\Gamma \vdash \Delta}
\infer1[WL]{\Gamma, P \vdash \Delta}
\end{prooftree}
\end{minipage}
\end{center}
\end{proposition}
\begin{proof} The semantics of propositional connectives is classical.
\end{proof}

\subsection{Some Proof Rules from \dL}
We show the soundness of the following proof rules from \dL because we will make use of them in our derivations.

\begin{proposition} The following proof rules of \dL (as developed in \cite{Pl2}) are sound for state formulas of \PdTL:
\begin{center}
 {\begin{prooftree} \hypo{\formulaSymbol}\infer1[G]{[\alpha]\formulaSymbol} \end{prooftree}} \hspace{5em}{\begin{prooftree} \hypo{\secondFormulaSymbol \to \formulaSymbol}\infer1[M]{[\alpha]\secondFormulaSymbol \to [\alpha]\formulaSymbol} \end{prooftree}}
\end{center}
\end{proposition}

\begin{proof}  We start with rule G. Assume that state formula $\formulaSymbol$ is valid; i.e. true in all states. Fix a state $\firstStateSymbol$. We must show $\firstStateSymbol \models [\alpha]\formulaSymbol$. By \rref{def:stateformulaSemantics}, this holds iff for every trace $\tracesymbol$ of $\alpha$ with $\text{first }\tracesymbol = \firstStateSymbol$, if $\val{\tracesymbol}{\formulaSymbol}$ is defined, then $\val{\tracesymbol}{\formulaSymbol}$ is true. By definition \rref{def:traceformulasemantics}, if $\val{\tracesymbol}{\formulaSymbol}$ is defined, then $\tracesymbol$ terminates and  $\val{\tracesymbol}{\formulaSymbol} = \val{\text{last }\tracesymbol, \formulaSymbol}$. Because $\formulaSymbol$ is valid, $\text{last }\tracesymbol \models \formulaSymbol$ and thus $ \val{\text{last }\tracesymbol, \formulaSymbol}$ is true. Therefore $\firstStateSymbol \models [\alpha]\formulaSymbol$ and so G is sound.

Next, for rule M, assume that state formula $\secondFormulaSymbol \to \formulaSymbol$ is valid. Fix a state $\firstStateSymbol$. We must show $\firstStateSymbol \models [\alpha]\secondFormulaSymbol \to [\alpha]\formulaSymbol$. So, assume  $\firstStateSymbol \models [\alpha]\secondFormulaSymbol$. We must show  $\firstStateSymbol \models [\alpha]\formulaSymbol$. By \rref{def:stateformulaSemantics}, this holds iff for every trace $\tracesymbol$ of $\alpha$ with $\text{first }\tracesymbol = \firstStateSymbol$, if $\val{\tracesymbol}{\formulaSymbol}$ is defined, then $\val{\tracesymbol}{\formulaSymbol}$ is true. By \rref{def:traceformulasemantics}, $\val{\tracesymbol}{\formulaSymbol}$ is defined iff $\tracesymbol$ terminates---in that case, $\val{\tracesymbol}{\secondFormulaSymbol}$ is also defined, and thus must be true since we have assumed $\firstStateSymbol \models [\alpha]\secondFormulaSymbol$. Since $\val{\tracesymbol}{\secondFormulaSymbol}$ is $ \val{\text{last }\tracesymbol}{\secondFormulaSymbol}$ by \rref{def:traceformulasemantics}, that means that $\text{last }\tracesymbol \models \secondFormulaSymbol$. But from our assumption that $\secondFormulaSymbol \to \formulaSymbol$ is valid, we have $\text{last }\tracesymbol \models \formulaSymbol$, and thus $\val{\tracesymbol}{\formulaSymbol}$ is true, as desired.
\end{proof}

\subsection{Soundness of $M_{\text{tae}}$}
\begin{proof}
We have the following prooftree.
The premises are at the top level of the tree and the conclusion is at the bottom level.
The only acceptable ``steps'' in prooftrees are via sound rules.
This has the effect that validity propagates downwards, so that when the premises are valid, the conclusion is also valid.

\begin{center}
\begin{prooftree}
\hypo{\secondFormulaSymbol \to \formulaSymbol}
\infer1[TopCl]{\cl{\secondFormulaSymbol} \to \cl{\formulaSymbol}}
\hypo{\secondFormulaSymbol \to \formulaSymbol}
\infer1[$G_{\text{tae}}$]{\rae{\alpha}{(\secondFormulaSymbol \to \formulaSymbol)}}
\infer2[$K_{\text{tae}}$]{\rae{\alpha}{\secondFormulaSymbol} \to \rae {\alpha}{\formulaSymbol}}
\end{prooftree}
\end{center}
 \end{proof}
 
 \subsection{Soundness of $\text{Comp}_{\text{tae}}$}
\begin{proof}
We construct the following prooftree, which uses the monotonicity property (M) from \dL:

\scalebox{.8}{
\begin{prooftree}
\hypo{\secondFormulaSymbol \to \rae{\alpha}{\formulaSymbol}}

\hypo{*}
\infer1[id]{\rae{\alpha}{\formulaSymbol} \to \rae{\alpha}{\formulaSymbol}}

\hypo{*}
\infer1[CGG]{\rae{\alpha}{\formulaSymbol} \to [\alpha]\cl{\phi}}

\hypo{\cl{\formulaSymbol} \to \rae{\beta}{\formulaSymbol}}
\infer1[M]{[\alpha]\cl{\formulaSymbol} \to [\alpha]\rae{\beta}{\formulaSymbol}}

\infer2[cut]{\rae{\alpha}{\formulaSymbol} \to [\alpha]\rae{\beta}{\formulaSymbol}}

\infer2[$\land R$]{\rae{\alpha}{\formulaSymbol}  \to \rae{\alpha}{\formulaSymbol} \land [\alpha]\rae{\beta}{\formulaSymbol}}

\infer2[cut]{\secondFormulaSymbol \to \rae{\alpha}{\formulaSymbol} \land [\alpha]\rae{\beta}{\formulaSymbol}}
\infer1[$[;]_{\text{tae}}$]{\secondFormulaSymbol \to \rae{\alpha; \beta}{\formulaSymbol}}

\end{prooftree}
}

 \end{proof}
Note that $\text{Comp}_{\text{tae}}$ is incomplete---consider, for example, a case in which $\formulaSymbol = \secondFormulaSymbol$ and $\alpha = \beta$ is an ODE for which both $\formulaSymbol$ and $\cl{\formulaSymbol} \setminus \formulaSymbol$ are invariant regions.  More concretely, say $\formulaSymbol$ is the interior of the circle of radius 1, i.e. $x^2 + y^2{<}1$ and say $\alpha$ is the vector field $(x', y') = (-y, x)$.  Then $\mrae{\secondFormulaSymbol}{\alpha; \beta}{\formulaSymbol}$, since $\formulaSymbol$ is an invariant region.  However, since $\cl{\formulaSymbol} \setminus \formulaSymbol$ is an invariant region, $\cl{\formulaSymbol} \to \rae {\beta}{\formulaSymbol}$ does not hold.
 
\subsection{Soundness of $\text{Ind}_{\text{tae}}$}
\begin{proof} 
We construct the following prooftree, which uses the G\"{o}del generalization rule (G) from \dL:
\[
\begin{prooftree}
\hypo{*}
\infer1[id]{\cl{\formulaSymbol} \vdash \cl{\formulaSymbol}}
\hypo{\cl{\formulaSymbol} \vdash \rae{\alpha}{\formulaSymbol}}
\infer1[$\to$ R]{\vdash \mrae{\cl{\formulaSymbol}}{\alpha}{\formulaSymbol}}
\infer1[G]{\cl{\formulaSymbol} \vdash [\alpha^*](\mrae{\cl{\formulaSymbol}}{\alpha}{\formulaSymbol})}
\infer2[$\land$ R]{\cl{\formulaSymbol} \vdash \cl{\formulaSymbol} \land [\alpha^*](\mrae{\cl{\formulaSymbol}}{\alpha}{\formulaSymbol})}
\infer1[$\text{I}_{\text{tae}}$]{\cl{\formulaSymbol} \vdash \rae{\alpha^*}{\formulaSymbol}}
\end{prooftree}
\]
 \end{proof}
 
\subsection{Soundness of $\text{loop}_{\text{tae}}$}
\begin{proof}
We construct the following prooftree (using WR and WL implicitly):
\[
\scalebox{.9}{
\begin{prooftree}
\hypo{\cl{\secondFormulaSymbol} \vdash \rae{\alpha}{\secondFormulaSymbol}}
\infer1[$\text{Ind}_{\text{tae}}$]{\cl{\secondFormulaSymbol} \vdash \rae{\alpha^*}{\secondFormulaSymbol}}
\infer1[$\to R$]{ \Gamma \vdash \cl{\secondFormulaSymbol} \to \rae{\alpha^*}{\secondFormulaSymbol}, \Delta}
\hypo{\secondFormulaSymbol \vdash \formulaSymbol}
\infer1[$\text{M}_{\text{tae}}$] {\rae{\alpha^*}{\secondFormulaSymbol} \to \rae{\alpha^*}{\formulaSymbol}}
\hypo{\Gamma \vdash \cl{\secondFormulaSymbol}, \Delta}
\infer2[$\to L$]{\Gamma, \cl{\secondFormulaSymbol} \to \rae{\alpha^*}{\secondFormulaSymbol} \vdash \rae{\alpha^*}{\formulaSymbol}, \Delta}
\infer2[cut]{\Gamma \vdash \rae{\alpha^*}{\formulaSymbol}, \Delta}
\end{prooftree}
}
\]
\end{proof}

\section{Motivating Example}\label{app:Example} So that the full proof of the motivating train example is contained in this appendix, we repeat some of the material found in \rref{sec:Example}.

Recall that we model the train example as:
\begin{align*} a = 0 \land v = 0 \to  & [\big(((?(v{<}100); a := 1) \cup (?(v = 100); a := -1)); \\
 & \hspace{2em} \{x' = v, v' = a \ \&\ 0{\leq}v{\leq}100\}\big)^*] \Box_{\text{tae}} v{<}100
\end{align*}
If the initial velocity and acceleration are both 0, then $v{<}100$ almost everywhere.
The train accelerates while $v{<}100$ and brakes if $v = 100$; it moves according to the system of differential equations $x' = v, v' = a$.
The evolution domain constraint $v{\leq}100$ indicates an event-triggered controller (see \cite{Pl2}).

We can use the structural rule ${\to}R$ and our induction proof rule $\text{loop}_{\text{tae}}$ with invariant $v{<}100$ to reduce this to showing $a = 0 \land v = 0 \vdash v{\leq}100$ (which holds by real arithmetic), $v{<}100 \vdash v{<}100$ (which holds identically), and 
\begin{align*}v{\leq}100 \vdash &[\big((?(v{<}100); a := 1) \cup (?(v = 100); a := -1)\big); \\
 & \hspace{2em} \{x' = v, v' = a \ \&\  0{\leq}v{\leq}100\}] \Box_{\text{tae}} v{<}100.
\end{align*}
We can then use axiom $[;]_{\text{tae}}$ to split this into the two main goals 
\begin{equation}\label{eqn:MG1}
 v{\leq}100 \vdash [(?(v{<}100); a := 1) \cup (?(v = 100); a := -1)]\Box_{\text{tae}} v{<}100
 \end{equation}
 and 
\begin{equation}\label{eqn:MG2}
\begin{split}
v{\leq}100 \vdash [(?(v{<}100); a := 1) &\cup (?(v = 100); a := -1)] \\
[ \{x' = v, v' = a 
& \&\  0{\leq}v{\leq}100\}]\Box_{\text{tae}} v{<}100.
\end{split}
\end{equation}
The main goal (\rref{eqn:MG1}) is relatively easy to prove. We can use $[\cup]_{\text{tae}}$ and $\land R$ to split the proof into the two subgoals 
\begin{equation}\label{eqn:SG1} v{\leq}100 \vdash [?(v{<}100); a := 1] \Box_{\text{tae}} v{<}100
\end{equation} and
\begin{equation}\label{eqn:SG2} v{\leq}100 \vdash [?(v = 100); a := -1] \Box_{\text{tae}} v{<}100. \end{equation}

Using $[;]_{\text{tae}}$, the first subgoal (\rref{eqn:SG1}) reduces to showing 
\begin{equation}\label{eqn:ez1}
v{\leq}100 \vdash [?(v{<}100)] \Box_{\text{tae}} v{<}100
\end{equation} and  
\begin{equation}\label{eqn:ez2}
v{\leq}100 \vdash [?(v{<}100)][a := 1] \Box_{\text{tae}} v{<}100.
\end{equation}
Using $[?]_{\text{tae}}$, (\rref{eqn:ez1}) holds iff $v{\leq}100  \vdash v{\leq}100$, which is identically true.
Using $[:=]_{\text{tae}}$,  (\rref{eqn:ez2}) holds iff  $v{\leq}100 \vdash [?(v{<}100)] ((v{\leq}100) \land [a := 1] v{<}100)$ holds, and this is a valid \dL formula. The second subgoal (\rref{eqn:SG2}) proves similarly.

The second main goal (\rref{eqn:MG2}) is more complicated because it involves ODEs reasoning.  We can use the \dL axiom $[\cup]$ and $\land R$ to split the proof into two subgoals:
\begin{equation} \label{eqn:SG3}
v{\leq}100 \vdash [?(v{<}100); a := 1] [ \{x' = v, v' = a\ \&\  0{\leq}v{\leq}100\}]\Box_{\text{tae}} v{<}100
\end{equation} and
\begin{equation}\label{eqn:SG4}
    v{\leq}100 \vdash [?(v = 100); a := -1] [ \{x' = v, v' = a\ \&\  0{\leq}v{\leq}100\}]\Box_{\text{tae}} v{<}100.
\end{equation}

Now, we can use the \dL axioms $[;]$, $[:=]$, and $[?]$ to reduce (\rref{eqn:SG3}) to 
\begin{equation} \label{eqn:ODE1}
v{\leq}100, v{<}100 \vdash [ \{x' = v, v' = 1\ \&\  0{\leq}v{\leq}100\}]\Box_{\text{tae}} v{<}100
\end{equation}

Similarly, we reduce (\rref{eqn:SG4}) to
\begin{equation}\label{eqn:ODE2}
    v{\leq}100, v{=}100 \vdash [ \{x' = v, v' = -1\ \&\  0{\leq}v{\leq}100\}]\Box_{\text{tae}} v{<}100.
\end{equation}

Now, we need to use axiom $['\&]_{\text{tae}}$, which says:
\begin{equation*}
    \rae{x' = f(x) \& \evDomainConstraint}{\FOLformulaSymbol} \leftrightarrow \cl{\FOLformulaSymbol} \land \forall t {>}0 \left(\left( \forall 0{\leq}s{\leq}t \ [x := y(s)]\evDomainConstraint \right) \to \secondFOLformulaSymbol \right)
\end{equation*}

First, we focus on (\rref{eqn:ODE1}).
For clarity, we use $v_0$ for the value of $v$ in the initial state before it starts evolving along the ODEs, and similarly we use $x_0$ for the value of $x$ in the initial state.
Following \rref{prop:definable}, $Q$ is $(1 = 0 \to t + v_0{<}100) \land t + v_0 {\leq}100$. 

So we reduce (\rref{eqn:ODE1}) to:
\begin{equation}\label{eqn:DL1}
\begin{split}
    & v_0{\leq}100, v_0{<}100 \vdash  v_0{\leq}100 \land \forall t {>}0 \\
    &\left(\left( \forall 0{\leq}s{\leq}t \ [x := .5s^2 + v_0s + x_0][v := s + v_0] 0{\leq}v{\leq}100 \right) \to Q \right)
    \end{split}
\end{equation}

Notice that (\rref{eqn:DL1}) is a \dL formula.
Since \PdTL is a conservative extension of \dL, we can use any and all relevant rules of \dL to close goal (\rref{eqn:DL1}).
We start by using the contextual equivalence rules of \dL \cite{Pl2} to replace $Q$ with (the logically equivalent formula) $t+v_0{\leq}100$.
Thus (\rref{eqn:DL1}) becomes
\begin{equation}\label{eqn:dLCEQ}
\begin{split}
    & v_0{\leq}100, v_0{<}100 \vdash  v_0{\leq}100 \land \forall t {>}0 \\
    &\left(\left( \forall 0{\leq}s{\leq}t \ [x := .5s^2 + v_0s + x_0][v := s + v_0] 0{\leq}v{\leq}100 \right) \to t+v_0{\leq}100\right)
    \end{split}
\end{equation}

Using the assignment axiom $[:=]$ from \dL, this reduces to
\begin{equation}\label{eqn:DLSub}
\begin{split}
   & v_0{\leq}100, v_0{<}100 \vdash v_0 \leq 100 \\
   & \land \forall t {>}0
   \left(\left( \forall 0{\leq}s{\leq}t (0{\leq}v_0 + s{\leq}100) \right) \to t+v_0{\leq}100\right)
       \end{split}
\end{equation}

The main thing we need to show (\rref{eqn:DLSub}) is that 
$$v_0{<}100 \vdash \forall t {>}0 \left(\left( \forall 0{\leq}s{\leq}t\ 0{\leq}v_0{+}s{\leq}100 \right) \to t+v_0{\leq}100 \right),$$
 and this is a valid sequent because in particular when $s = t$, $0{\leq}v_0 + t{\leq}100$.

In a similar way, to close (\rref{eqn:ODE2}), using $['\&]_{\text{tae}}$ gives:
\begin{equation}\label{eqn:DL2}
\begin{split}
    & v_0{\leq}100, v_0{=}100 \vdash  v_0{\leq}100 \land \forall t {>}0 \\
    &\left(\left( \forall 0{\leq}s{\leq}t \ [x := -.5s^2 + v_0s + x_0][v := -s + v_0] 0{\leq}v{\leq}100 \right) \to -t + v_0 {\leq}100 \right)
    \end{split}
\end{equation}

Using the assignment axiom from \dL, this reduces to
\begin{equation*}
    v_0{\leq}100 \land v_0 = 100 \vdash v_0{\leq}100 \land \forall t {>}0 \left(\left( \forall 0{\leq}s{\leq}t 0{\leq}-s + v_0{\leq}100 \right) \to -t + v_0 {\leq}100 \right)
\end{equation*}
which is a valid sequent.
\end{document}